\documentclass[11pt,reqno]{amsart}
\usepackage{graphicx,amsmath,verbatim,amssymb,lineno,eucal,titletoc, epsf, epsfig, complexity}
\usepackage{bbm}
\usepackage{hyperref}
\usepackage{tikz}
\usepackage{enumerate}
\usetikzlibrary{arrows}
\usetikzlibrary{positioning}
\usetikzlibrary{shapes}
\usepackage{url}

\expandafter\let\csname ver@amsthm.sty\endcsname\relax
\let\theoremstyle\relax

\let\qedhere\relax

\usepackage{hyperref}
\usepackage{amsthm}
\usepackage{cleveref}

\newcommand{\arxiv}[1]{{\tt \href{http://arxiv.org/abs/#1}{arXiv:#1}}}

\newcommand{\floor}[1]{\left\lfloor {#1} \right\rfloor}
\newcommand{\ceiling}[1]{\left\lceil {#1} \right\rceil}

\newcommand{\old}[1]{}
\newcommand{\moniker}[1]{{\em (#1)}}

\newcommand{\cpic}[1]{\parbox[c][0.65in]{1.4in}
{\centering #1}}
\newcommand{\ind}{\one}

\newcommand{\df}{\textbf}

\newtheorem{theorem}{Theorem}[section]
\newtheorem{prop}[theorem]{Proposition}

\newtheorem{lemma}[theorem]{Lemma}
\newtheorem{corollary}[theorem]{Corollary}

\theoremstyle{remark}
\newtheorem*{remark}{Remark}

\crefname{maintheorem}{Theorem}{Theorems}
\crefname{theorem}{Theorem}{Theorems}
\crefname{lemma}{Lemma}{Lemmas}
\crefname{prop}{Proposition}{Propositions}
\crefname{corollary}{Corollary}{Corollaries}
\crefname{section}{Section}{Sections}
\crefname{figure}{Figure}{Figures}
\crefname{table}{Table}{Tables}
\crefname{definition}{Definition}{Definitions}

\numberwithin{counter}{section}

\theoremstyle{definition}
\newtheorem{definition}[theorem]{Definition}

\def\bb{\mathbf{b}}
\def\cc{\mathbf{c}}
\def\dd{\mathbf{d}}
\def\mm{\mathbf{m}}

\def\qq{\mathbf{q}}
\def\rr{\mathbf{r}}

\def\uu{\mathbf{u}}
\def\vv{\mathbf{v}}
\def\ww{\mathbf{w}}
\def\xx{\mathbf{x}}
\def\yy{\mathbf{y}}
\def\zz{\mathbf{z}}
\def\00{\mathbf{0}}
\def\11{\mathbf{1}}
\def\mm{\mathbf{m}}

\def\til{\widetilde}
\def\un{\Box}
\def\start{0}
\def\Add{{\tt adder}}
\def\Spl{{\tt splitter}}
\def\Topp{{\tt toppler}}
\def\Del{{\tt delayer}}
\def\Pre{{\tt presink}}
\def\Adds{{\tt adders}}
\def\Spls{{\tt splitters}}
\def\Topps{{\tt topplers}}
\def\Dels{{\tt delayers}}
\def\Pres{{\tt presinks}}

\def\ld{\boldsymbol{\lambda}}

\def\Proc{\CMcal{P}}  
\def\ProcA{\CMcal{A}}
\def\ProcQ{\CMcal{Q}}
\def\ProcR{\CMcal{R}}
\def\ProcS{\CMcal{S}}
\def\ProcT{\CMcal{T}}

\def\ProcI{\CMcal{I}}
\def\ProcJ{\CMcal{J}}
\def\ProcG{\CMcal{G}}
\def\ProcM{\CMcal{M}}
\def\ProcH{\CMcal{H}}
\def\ProcD{\CMcal{D}}
\def\ProcU{\CMcal{U}}
\def\Net{{ \mathcal{N}}}  

\def\zero{\mathbf{0}}
\def\one{\mathbbm{1}}

\def\basis{\mathbf{e}}

\def\N{\mathbb{N}}
\def\Z{\mathbb{Z}}
\def\Q{\mathbb{Q}}

\newcommand{\pre}{\widehat{M}}
\renewcommand{\K}{\mathsf{L}}

\pgfdeclareshape{spiral}{
	\nodeparts{}
 \inheritsavedanchors[from={circle}]
  \inheritbackgroundpath[from={circle}]
  \inheritanchorborder[from={circle}]
  \foreach \x in {center,north east,north west,north,south,south east,south west,east,west}{
    \inheritanchor[from={circle}]{\x}
  }
	\behindbackgroundpath{
		\draw [ultra thick,blue,domain=-.6:.48,variable=\t,
smooth,tension=.65,samples=30,line cap=round]
        plot ({(-15*\t+4.3) r}: {.6*(-\t*\t*\t+\t+.4)});
	}
}

\pgfdeclareshape{plus}{
	\nodeparts{}
 \inheritsavedanchors[from={circle}]
  \inheritbackgroundpath[from={circle}]
  \inheritanchorborder[from={circle}]
  \foreach \x in {center,north east,north west,north,south,south east,south west,east,west}{
    \inheritanchor[from={circle}]{\x}
  }
	\behindbackgroundpath{
		\draw [very thick]
        (-.125,0) -- (.125,0) (0,-.125) -- (0,.125);
	}
}

\tikzstyle{abelian}=[>=angle 90,auto,node
distance=1.8cm,thick] \tikzstyle{contour}=[thick]

\tikzstyle{adder}=[plus,inner sep=0pt,minimum
size=3mm,outer sep=1pt]

\tikzstyle{symbol}=[circle,inner sep=0pt]
\tikzstyle{joint}=[circle,fill=black,minimum size=2mm,inner
sep=1pt]
 \tikzstyle{splitter}=[circle,fill=black,minimum
size=2.2mm,inner sep=0pt,outer sep=1pt]
 \tikzstyle{toppler}=[draw,circular sector,circular sector angle=65,shape
border rotate=90,fill=yellow,minimum height=9mm,inner
sep=1pt,font=\boldmath,outer sep=2pt]
 \tikzstyle{presink}=[spiral,minimum size=9mm,inner sep=0pt, outer sep=2pt]

 \tikzstyle{delayer}=[draw,regular polygon,regular polygon
sides=8,minimum size=9mm,fill=red!70,inner sep=3pt,outer
sep=1.5pt]

\tikzstyle{block}=[draw,fill=black!10,thick,rectangle]

\tikzstyle{dot}=[circle,fill=black,minimum size=1mm,inner
sep=1.7pt]
\tikzstyle{circ}=[circle,draw,fill=white,minimum size=1mm,inner
sep=1.7pt]
\tikzstyle{hili}=[line width=3pt,red]

\tikzset{prime/.style={label={[label distance=-2mm]155:#1}}}

\begin{document}

\title{Abelian logic gates}

\author[Holroyd]{Alexander E. Holroyd}
\address{Alexander E. Holroyd, Microsoft Research, Redmond, WA 98052, USA.
\newline
{\tt \url{http://research.microsoft.com/~holroyd}}}
\author[Levine]{Lionel Levine}
\address{Lionel Levine, Cornell University, Ithaca, NY 14853,
USA. \newline {\tt {\url{http://www.math.cornell.edu/~levine}}}}
\author[Winkler]{Peter Winkler}
\address{Peter Winkler, Dartmouth College, Hanover, NH 03755, USA.
\newline
  {\tt \url{http://math.dartmouth.edu/~pw}}}
\thanks{The second author is supported by NSF grant
\href{http://www.nsf.gov/awardsearch/showAward?AWD_ID=1455272}{DMS-1455272} and a Sloan Fellowship. The third author is supported by NSF
grant DMS-1162172.}

\date{Revised April 9, 2018}

\begin{abstract}
An abelian processor is an automaton whose output is independent of the order
of its inputs.  Bond and Levine have proved that a network of abelian
processors performs the same computation regardless of processing order
(subject only to a halting condition).  We prove that any finite abelian
processor can be emulated by a network of certain very simple abelian
processors, which we call gates. The most fundamental gate is a {\em
toppler}, which absorbs input particles until their number exceeds some given
threshold, at which point it topples, emitting one particle and returning to
its initial state. With the exception of an {\em adder} gate, which simply
combines two streams of particles, each of our gates has only one input wire,
which sends letters (``particles") from a unary alphabet.
Our results can be reformulated in terms of the functions computed by
processors, and one consequence is that any increasing function from $\N^k$
to $\N^\ell$ that is the sum of a linear function and a periodic function can
be expressed in terms of (possibly nested) sums of floors of quotients by integers.
\end{abstract}
\maketitle

\section{Introduction}
\label{s.intro}

Consider a network of finite-state automata, each with a finite input and
output alphabet. What can such a network reliably compute if the wires
connecting its components are subject to unpredictable delays? The networks
we will consider have a finite set of $k$ input wires and $\ell$ output
wires.
These wires are unary (each carries letters from a $1$-letter alphabet), and even they
are subject to delays, so the network computes a function
$\N^k \to \N^\ell$: The input is a $k$-tuple of natural numbers ($\N =
\{0,1,2,\ldots\}$) indicating how many letters are fed along each input wire,
and the output is an $\ell$-tuple indicating how many letters are emitted
along each output wire.

The essential issue such a network must overcome is that the order in which
input letters arrive at a node must not affect the output. To address this
issue, Bond and Levine \cite{BL15a}, following Dhar \cite{Dha99a, Dha06},
proposed the class of \emph{abelian networks}.  These are networks each of
whose components is a special type of finite automaton called an
\emph{abelian processor}.

Certain abelian networks such as sandpile
\cite{Ost03,SD12} and rotor \cite{LP09,HP10,FL13} networks produce intricate fractal
outputs from a simple input.  Abelian networks can be used to solve certain integer programs asynchronously \cite{BL15a} and to detect graph planarity \cite{CCG14}. From the point of view of computational
complexity, predicting the final state of a sandpile on a finite simple graph can be done in polynomial time \cite{Tar88}, and in fact this problem is $\P$-complete \cite{MN99}.
But on finite directed multigraphs, deciding whether a sandpile halts 
is already $\NP$-complete \cite{FL15}. 
(Whether a sandpile halts is independent of the order of topplings.)
For further complexity results, see \cite{MM09,MM11,CL13,HKT15,KT15,PP15}.
%
Analogous problems on infinite graphs are undecidable:
An abelian network whose underlying graph is $\Z^2$, or a sandpile network whose underlying graph is the product of $\Z^2$ with a finite path, can emulate a Turing machine \cite{Cai15}.

The following definition is equivalent to that in \cite{BL15a} but simpler to
check.
A \df{processor} with input alphabet $A$, 
output alphabet $B$ and state space $Q$ is a collection of \df{transition maps} and \df{output maps}
	\[ t_i : Q \to Q  \quad\text{and}\quad o_i : Q \to \N^B \]
 indexed by  $i \in A$.
The processor is \df{abelian} if
\begin{equation}\label{commute}
 t_i t_j = t_j t_i \quad\text{and}\quad
			       o_i + o_j t_i = o_j + o_i t_j
\end{equation}
for all $i,j \in A$. The interpretation is that if the processor receives
input letter $i$ while in state $q$, then it transitions to state $t_i(q)$
and outputs $o_i(q)$.  The first equation in \eqref{commute} above asserts
that the processor moves to the same state after receiving two letters,
regardless of their order. The second guarantees that it produces the same
output.  The processor is called \df{finite} if both the alphabets $A$, $B$
and the state space $Q$ are finite.  In this paper, all abelian processors
are assumed to be finite and to come with a distinguished starting state
$q^\start$ that can access all states: each $q \in Q$ can be obtained by a
composition of a finite sequence of transition maps $t_i$ applied to
$q^\start$.

We say that an abelian processor \df{computes} the function
	$ F : \N^A \to \N^B $ if inputting $\xx_a$ letters $a$ for each $a \in A$
results in the output of $(F(\xx))_b$ letters $b$ for each $b \in B$. Our
convention that the various inputs and outputs are represented by different
letters is useful for notational purposes.  An alternative viewpoint would be
to regard all inputs and outputs as consisting of indistinguishable
``particles'', whose roles are determined by which input or output wire they
pass along.

An \df{abelian network} is a directed graph with an abelian processor located
at each node, with outputs feeding into inputs according to the graph
structure, and some inputs and outputs designated as input and output wires
for the entire network.  (We give a more formal definition below in
\textsection\ref{s.network}.) An abelian network can compute a function as
follows.  We start by feeding some number of letters along each input wire.
Then, at each step, we choose any processor that has at least one letter
waiting at one of its inputs, and process that letter, resulting in a new
state of that processor, and perhaps some  letters emitted from its outputs.
If after finitely many steps all remaining letters are located on the output
wires of the network, then we say that the computation halts.

The following is a central result of \cite{BL15a},
generalizing the ``abelian property'' of Dhar~\cite{Dha90}
and Diaconis and Fulton \cite[Theorem~4.1]{DF91} (see
\cite{notes} for further background).  Provided the
computation halts, it does so
regardless of the choice of processing order.  Moreover,
the letters on the output wires and the final states of the processors are also
independent of the processing order.  Thus, a network that halts on
all inputs computes a function from $\N^k$ to $\N^\ell$
(where $k$ and $\ell$ are the numbers of input and output
wires respectively).  This function is itself of a form
that can be computed by some abelian processor, and we say
that the network \df{emulates} this processor.

The main goal of this paper is to prove a result in the opposite direction.
Just as any boolean function $\{0,1\}^A \to \{0,1\}^B$ can be
computed by a circuit of AND, OR and NOT gates, we show that any function $\N^A\to\N^B$
computed by an abelian processor can be computed by a network of simple
\emph{abelian logic gates}, specified below.  Furthermore (as in the boolean case), the
network can be made \df{directed acyclic}, which is to say that the graph has
no directed cycles. 

For example, a sandpile or rotor-router process (see, e.g.~\cite{notes,HP10})
defined on a finite graph can be thought of as a computer whose input is particles
placed on vertices, and whose output constitutes particles collected at designated points.
It is thus an abelian processor, and our theorems show that it can be emulated by a directed acyclic network of simple abelian gates.

\begin{theorem}
\label{t.eventually.intro} Any finite abelian processor can be emulated by a
finite directed acyclic network of \Adds, \Spls, \Topps, \Dels\,and \Pres.
\end{theorem}

If the processor satisfies certain additional conditions, then some gates are
not needed.  An abelian processor is called \df{bounded} if the range of the
function that it computes is a finite subset of $\N^B$.

\begin{theorem}
\label{t.bounded.intro}
Any {\em bounded} finite abelian processor can
be emulated by a finite directed acyclic network of \Adds, \Spls, \Dels\,and \Pres.
\end{theorem}

An abelian processor $\Proc$ is called \df{recurrent} if for every pair of
states $q,q'$ there is a finite sequence of input letters that causes it to
transition from $q$ to $q'$.  An abelian processor that is not recurrent is
called \df{transient}.

\begin{theorem}
\label{t.recurrent.intro} Any {\em recurrent} finite abelian processor can be
emulated by a finite directed acyclic network of \Adds, \Spls\,and \Topps.
\end{theorem}

\subsection{The gates}
\label{s.gates}

\begin{table}
\begin{tabular}{|c|c|c|}
\hline
\multicolumn{3}{|c|}{\bf Single state} \\
  \hline
\Add &
\cpic{
\begin{tikzpicture}[abelian]
  \node[adder] (a) {};
  \draw (a) edge[->] ++(1,0);
  \draw (a) edge[<-] ++(-1,-.6);
  \draw (a) edge[<-] ++(-1,.6);
\end{tikzpicture}
}
& $(x,y)\mapsto x+y$ \\
\hline
\Spl &
\cpic{\begin{tikzpicture}[abelian]
  \node[splitter] (a) {};
  \draw (a) edge[->] ++(1,.6);
  \draw (a) edge[->] ++(1,-.6);
  \draw (a) edge[<-] ++(-1,0);
\end{tikzpicture}}
& $x\mapsto (x,x)$  \\
\hline\hline
\multicolumn{3}{|c|}{\bf Recurrent} \\
\hline
\Topp\ $(\lambda \geq 2$) & \cpic{\begin{tikzpicture}[abelian]
  \node[toppler] (a) {$\lambda$};
  \draw[<-] (a.west) -- ++(-1,0);
  \draw[->] (a.east) -- ++(1,0);
\end{tikzpicture}}
& $\displaystyle x\mapsto \Bigl\lfloor \frac x\lambda \Bigr\rfloor$  \\
\hline
primed \Topp\ ($1 \leq q < \lambda$) &
\cpic{\begin{tikzpicture}[abelian]
  \node[toppler,prime={$q$}] (a) {$\lambda$};
  \draw[<-] (a.west) -- ++(-1,0);
  \draw[->] (a.east) -- ++(1,0);
\end{tikzpicture}}
& $\displaystyle x\mapsto \Bigl\lfloor \frac {x+q}\lambda \Bigr\rfloor$  \\
\hline \hline
\multicolumn{3}{|c|}{\bf Transient} \\
\hline
\Del &
\cpic{\begin{tikzpicture}[abelian]
  \node[delayer] (a) {};
  \draw[<-] (a.west) -- ++(-1,0);
  \draw[->] (a.east) -- ++(1,0);
\end{tikzpicture}}
& $\begin{aligned}x&\mapsto \max(x-1,0)\end{aligned}$  \\
\hline
\Pre &
\cpic{\begin{tikzpicture}[abelian]
  \node[presink] (a) {};
  \draw[<-] (a.west) -- ++(-1,0);
  \draw[->] (a.east) -- ++(1,0);
\end{tikzpicture}}
& $\begin{aligned}x&\mapsto \min(x,1)\end{aligned}$  \\
\hline
\end{tabular}
\smallskip
\caption{Abelian gates and the functions they compute.}\label{fig:gates}
\end{table}
Table~\ref{fig:gates} lists our abelian logic gates, along
with the symbols we will use when illustrating networks. A
\df{splitter} has one incoming edge, two outgoing edges,
and a single internal state.  When it receives a letter, it
sends one letter along each outgoing edge.  On the other
hand, an \df{adder} has two incoming edges, one outgoing
edge, and again a single internal state.  For each letter
received on either input, it emits one letter.  The rest of
our gates each have just one input and one output.

For integer $\lambda \geq 2$, a $\lambda$-\df{toppler} has
internal states $0,1,\ldots,\lambda{-}1$.  If it receives a
letter while in state $q<\lambda{-}1$, it transitions to
state $q{+}1$ and sends nothing.  If it receives a letter
while in state $\lambda{-}1$, it ``topples": it transitions
to state $0$ and emits one letter.  A $\lambda$-toppler
that begins in state $0$ computes the function $x \mapsto
\floor{x/\lambda}$; if begun in state $q>0$ it computes the
function $x \mapsto \floor{(x{+}q)/\lambda}$.  A toppler is
called \df{unprimed} if its initial state is $0$, and
\df{primed} otherwise.

The above gates are all recurrent. Finally, we have two transient gates whose
behaviors are complementary to one another.  A \df{delayer} has two internal
states $0,1$.  If it receives an input letter while in state $0$, it moves
permanently to state $1$, emitting nothing.  In state $1$ it sends out one
letter for every letter it receives. Thus, begun it state $0$, it computes
the function $x \mapsto \max(x{-}1,0)=(x-1)^+$. A \df{presink} has two
internal states $0,1$.  If it receives a letter while in state $0$, it
transitions permanently to state $1$ and emits one letter.  All subsequent
inputs are ignored.  From initial state $0$ it computes $x \mapsto
\min(x,1)=\ind[x>0]$.

The topplers form an infinite family indexed by the
parameter $\lambda \geq 2$.  If we allow our network to
have feedback (i.e., drop the requirement that it be
directed acyclic) then we need only the case $\lambda=2$,
and in particular our palette of gates is reduced to a
finite set.  Feedback also allows us to eliminate one
further gate, the delayer.

\begin{prop}
\label{p.feedback}
 For any $\lambda\geq 3$, a $\lambda$-\Topp\,can
 be emulated by a finite abelian network of
  \Adds, \Spls\,and $2$-\Topps.  So can a \Del.
\end{prop}

The toppler is a very close relative of the two most extensively studied
abelian processors: the sandpile node and the rotor node.
Specifically, for a node of degree $k$, either of
these is easily emulated by $k$ suitably primed topplers in parallel, as in
\cref{f-rotor}. (Sandpiles and rotors are typically considered on undirected
graphs, in which case the $k$ inputs and $k$ outputs of a vertex are both
routed along its $k$ incident edges). Rotor aggregation \cite{LP09} can also
be emulated, by inserting a delayer into the network for the rotor.
\begin{figure}
\centering
\begin{tikzpicture}[abelian]
\node[adder](a){};
\node[splitter,right= 1cm of a](s){};
\node[toppler,right of=s,prime=1](t2){$3$};
\node[toppler,above of=t2,prime=0](t1){$3$};
\node[toppler,below of=t2,prime=2](t3){$3$};
\draw[<-] (a)--++(-1,0);
\draw[<-] (a)--++(-1,-1);
\draw[<-] (a)--++(-1,1);
\draw[->] (a)--(s);
\draw[->] (s) edge[bend left=20](t1.west);
\draw[->] (s) edge(t2);
\draw[->] (s) edge[bend right=20](t3.west);
\draw[->] (t1)--++(1,0);
\draw[->] (t2)--++(1,0);
\draw[->] (t3)--++(1,0);
\end{tikzpicture}
\caption{Emulating a rotor of degree $3$ with topplers.
To emulate a sandpile node, prime the three topplers identically
(and optionally combine them into one toppler preceding a splitter).
For a rotor aggregation node, insert a
delayer between the adder and the splitter.
}\label{f-rotor}
\end{figure}
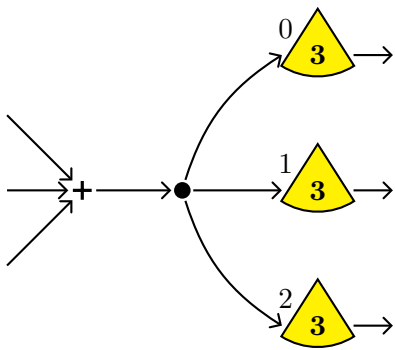

\subsection{Unary input}
\label{unary-intro}

A processor has \df{unary input} if its input alphabet $A$
has size $1$ (so that it computes a function
$\N\to\N^\ell$).  It is easy to see from the definition
that \emph{any} finite-state processor with unary input is
automatically abelian.  Indeed, the same holds for any
processor with \emph{exchangeable} inputs, i.e.\ one whose
transition maps and output maps are identical for each
input letter.  (Such a processor can be emulated by adding
all its inputs and feeding them into a unary-input
processor).  Note that all our gates have unary input
except for the adder, which has exchangeable inputs.

Theorems~\ref{t.eventually.intro}--\ref{t.recurrent.intro}
have rather straightforward proofs if we restrict to unary-input
processors.  (See Lemmas~\ref{unary-recurrent} and
\ref{unary-trans}.) Our main contribution is that
unary-input gates (and adders) suffice to emulate processors
with any number of inputs. (In contrast, elementary
considerations will show that there is no loss of
generality in restricting to processors with unary
\emph{output}; see Lemma~\ref{l.unaryoutput}.)

\subsection{Function classes}
\label{s.classes}

\begin{figure}
\includegraphics[width=.49\textwidth]{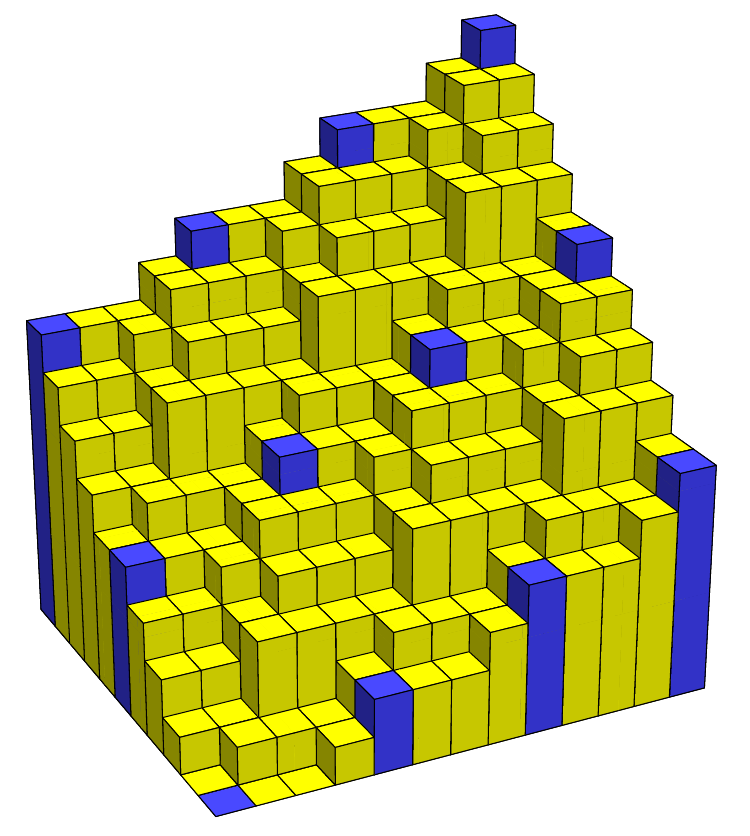}\hfill
\includegraphics[width=.49\textwidth]{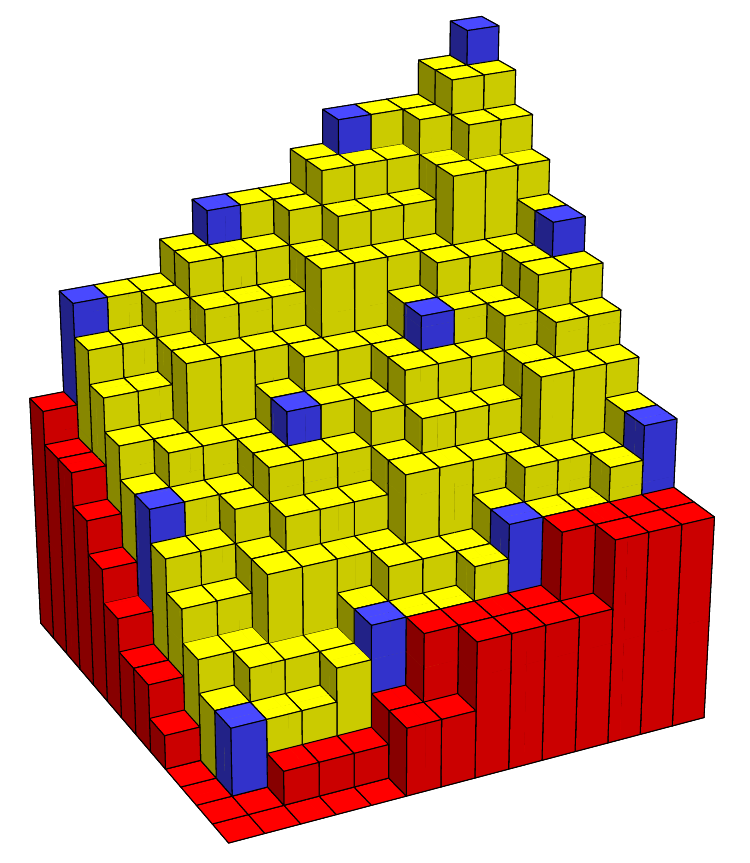}
\caption{\emph{Left:} the graph of a ZILP function $f:\N^2\to\N$.  The height of a bar
gives the value of the function, and the origin is at the front of the
picture.  The periodic component has periods $4$ and $5$ in the two
coordinates, as indicated by the highlighted bars; the linear part has slopes $2/4$
and $4/5$ respectively.  \emph{Right:} A ZILEP function comprising the same ``recurrent part"
together with added ``transient margins."}
\label{bar-graph}
\end{figure}
An important preliminary step in the proofs of
\cref{t.eventually.intro,t.bounded.intro,t.recurrent.intro} will be to
characterize the functions that can be computed by abelian processors (as
well as by the bounded and recurrent varieties). The characterizations turn
out to be quite simple.  A function $F:\N^k\to\N^\ell$ is computed by some
finite abelian processor if and only if: (i) it maps the zero vector
$\zero\in\N^k$ to $\zero\in\N^\ell$; (ii) it is (weakly) increasing; and
(iii) it can be expressed as a linear function plus an \emph{eventually
periodic} function (see \cref{ep-def} for precise meanings).  We call a
function satisfying (i)--(iii) \df{ZILEP} (zero at zero, increasing, linear
plus eventually periodic).  On the other hand, a function is computed by some
\emph{recurrent} finite abelian processor if it is \df{ZILP}:
\emph{eventually periodic} is replaced with \emph{periodic}. \cref{bar-graph}
shows examples of ZILP and ZILEP functions of two variables, illustrating
some of the difficulties to be overcome in computing them by networks.

Our main theorems may be recast in terms of functions rather than processors.
Table~\ref{table.summary} summarizes our main results from this perspective.
For instance, the following is a straightforward consequence of \cref{t.recurrent.intro}.
\begin{samepage}
\begin{corollary}\label{t.functions}
\moniker{Recurrent abelian functions} Let $\mathcal{R}$ be the smallest set
of functions $F:\N^k \to \N$ containing the constant function $1$ and the
coordinate functions $x_1, \ldots, x_k$, and closed under addition and
compositions of the form $F \mapsto \floor{F/\lambda}$ for integer $\lambda
\geq 2$.  Then $\mathcal{R}$ is the set of all increasing functions $\N^k \to
\N$ expressible as $L+P$ where $L, P : \N^k \to \Q$ with $L$ linear and $P$
periodic.
\end{corollary}
\end{samepage}
	\begin{table}
	\begin{tabular}{|c|ccc|c|}
\hline
	Components & L & + & P & Theorem(s) \\
	\hline
	(splitter and adder only) &  linear & + & zero & Lemma~\ref{l.linear} \\
	presink, delayer &  linear & + & eventually constant & \ref{t.bounded.intro}, \ref{t.immutable} \\
	$\lambda$-toppler & linear & + & periodic & \ref{t.recurrent.intro}, \ref{t.recurrent} \\
	$\lambda$-toppler, presink, delayer & linear & + & eventually periodic & \ref{t.eventually.intro}, \ref{t.eventually} \\
	\hline
	\end{tabular}
	\medskip
	\caption{Four different classes of abelian network.
	The second column indicates the class of increasing functions $L+P : \N^k \to \N^\ell$ computable by a finite, directed acyclic abelian network whose components are splitters, adders and the gates listed in the first column. In the first two lines $L$ and $P$ take values in $\N^\ell$, while in the last two lines $L$ and $P$ take values in $\Q^\ell$.}
	\label{table.summary}
	\end{table}

\subsection{Outline of article}

Section~\ref{s.function} identifies the classes of functions computable by
abelian processors, as described above, and formalizes the definitions and
claims relating to abelian networks.  Section~\ref{s.basic} contains a few
elementary reductions including the proof of Proposition~\ref{p.feedback}.
The core of the paper is Sections~\ref{s.recurrent}, \ref{s.bounded}
and~\ref{s.transient}, which are devoted respectively to the proofs of
Theorems~\ref{t.recurrent.intro}, \ref{t.bounded.intro}
and~\ref{t.eventually.intro}.  These proofs are by induction on the number of
inputs to the processor; \cref{t.eventually.intro} is by far the hardest. A
recurring theme in the proofs is \emph{meagerization}, which amounts to use
of the easily verified identity
	\begin{equation} \label{e.meagerization} x = \floor{\frac{x}{m}} +
\floor{\frac{x+1}{m}} + \cdots + \floor{\frac{x+m-1}{m}} \end{equation} for
positive integers $x$ and $m$.  A key step in the proof of the general
emulation result, Theorem~\ref{t.eventually.intro}, is the introduction of a
ZILP function that computes the minimum of its $n$ arguments provided they
are not too far apart (\cref{p.blackbox}). That this function in turn can be
emulated follows from the recurrent case, Theorem~\ref{t.recurrent.intro}.

In Section~\ref{s.lowerbounds}  we show that no gates can
be omitted from our list.  We conclude by posing some open
problems in Section~\ref{s.open}.

\section{Functions computed by abelian processors and networks}
\label{s.function}

In preparation for the proofs of the main results about emulation, we begin
by identifying the classes of functions that need to be computed.

\subsection{Abelian processors}

If $\Proc$ is an abelian processor with input alphabet $A$, and $w = i_1
\cdots i_\ell$ is a word with letters in $A$, then we define the transition
and output maps corresponding to the word:
	\[ t_w := t_{i_\ell} \cdots t_{i_1}; \]
	\[ o_w := o_{i_1} + o_{i_2} t_{i_1} + o_{i_3}t_{i_2}t_{i_1}
+ \cdots + o_{i_\ell} t_{i_{\ell-1}} \cdots t_{i_1}. \]

\begin{lemma}
For any words $w,w'$ such that $w'$ is a permutation of $w$, we have $t_w =
t_{w'}$ and
	$o_w = o_{w'}$.
\end{lemma}
\begin{proof}
This follows from the definition of an abelian processor, by induction on the length of $w$.
\end{proof}

The function $f=f_{\Proc}$ computed by an abelian processor $\Proc$ with
initial state $q^\start$ is given by
	\[ f(\xx) = o_{w(\xx)}(q^\start),\qquad \xx\in\N^A, \] where $w(\xx)$ is
any word that contains $x_i$ copies of the letter $i$ for all $i \in A$.  We
denote vectors by boldface lower-case letters, and their coordinates by the
corresponding lightface letter, subscripted.

\begin{lemma}
\label{l.parallelogram} Let $f=f_{\Proc}$. If $t_\yy(q^\start) =
t_{\yy'}(q^\start)$ then for any $\xx,\yy,\yy' \in \N^k$
	\[ f(\xx+\yy) - f(\yy) = f(\xx+\yy') - f(\yy'). \]
\end{lemma}

\begin{proof}
Let $t_\yy(q^\start) = t_{\yy'}(q^\start) = q$; then $f(\xx+\yy) = f(\yy) + o_{w(\xx)}(q)$
and $f(\xx+\yy) = f(\yy') + o_{w(\xx)}(q)$. In other words,
since $\yy$ and $\yy'$ leave $\Proc$ in the same state $q$, the subsequent input of $\xx$ has the same effect.
\end{proof}

\begin{definition} \label{ep-def} A function
$f : \N^k \to \N^\ell$ is (weakly) \df{increasing} if $\xx \leq \yy$ implies
$f(\xx) \leq f(\yy)$ where $\leq$ denotes the coordinatewise partial
ordering. A function $P : \N^k \to \Q^\ell$ is \df{periodic} if there is a
subgroup $\Lambda \subset \Z^k$ of finite index such that $P(\xx) = P(\yy)$
whenever $\xx-\yy \in \Lambda$. A function $P$ is \df{eventually periodic} if
there exist $\lambda_1, \ldots, \lambda_k \geq 1$ and $r_1,\ldots,r_k$ such
that for each $i=1,\ldots,k$, if $\xx \geq r_i \basis_i$ then
	\begin{equation} \label{e.eventuallyperiodicdef} P(\xx) = P(\xx+\lambda_i
\basis_i). \end{equation}
Here $\basis_i$ is the $i$th standard basis vector.
\end{definition}

Note that given $\ld$ and $\rr$, an eventually periodic function is determined by its values on the box $[0,r_1+\lambda_1] \times \cdots \times [0,r_k+\lambda_k]$. This notion of eventually periodic is intermediate in strength. A stronger requirement would be that $P$ agrees with some periodic function outside a finite set. A weaker requirement would be that \eqref{e.eventuallyperiodicdef} holds only for $\xx \geq \rr$.

An eventually periodic function is not generally periodic 
outside any finite box.  The reason
is that there are typically infinitely many grid points $\xx$ (that would, e.g., be red in \cref{bar-graph}) which satisfy
$\xx \geq r_i \basis_i$ for some but not all $i$.

\begin{samepage}
\begin{theorem}\label{t.eventually}
Let $k, \ell \geq 1$. A function $f: \N^k \to \N^\ell$ can be computed by a
finite abelian processor if and only if $f$ satisfies all of the following.
	\begin{enumerate}[{\em (i)}]
	\item $f(\zero)=0$.
	\item $f$ is increasing.
	\item $f = L+P$ for a linear function $L$ and an eventually periodic function $P$.
	\end{enumerate}
\end{theorem}
\end{samepage}
As mentioned earlier, we call a function satisfying (i)--(iii) \df{ZILEP}.
 Note that $L$ and $P$ need not have integer-valued coordinates, and $P$ need
not be nonnegative. For example, the function $f(x) = \floor{x/2}$ is ZILEP
(in fact ZILP) with $L(x)=x/2$ and $P(x) = -\ind [x \text{ is odd} ]/2$.

\begin{proof}[Proof of \cref{t.eventually}]
Any $f = f_{\Proc}$ trivially satisfies $f(\zero)=0$. To see that $f$ is
increasing, given $\xx \leq \yy$ there are words $w(\xx)$ and $w(\yy)$ (where
the number of occurrences of letter $i$ in $w(\zz)$ is $z_i$) for which
$w(\xx)$ is a prefix of $w(\yy)$.  Then
	\[ o_{w(\xx)} + o_{u} t_{w(\xx)} = o_{w(\yy)}. \]
Since the second term on the left is nonnegative, $o_{w(\xx)}(q) \leq o_{w(\yy)}(q)$.

To prove (iii), note that since $Q$ is finite, some power of $t_i$ is idempotent, that is
	\[ t_i^{2\lambda_i} = t_i^{\lambda_i} \]
for some $\lambda_i \geq 1$. Let $L : \N^k \to \N^\ell$ be the linear function sending
	\[ \lambda_i \basis_i \mapsto  f(2\lambda_i \basis_i) - f(\lambda_i
\basis_i) \] for each $i=1,\ldots,k$.   Now we apply Lemma~\ref{l.parallelogram} with
$\yy=\lambda_i \basis_i$ and $\yy'=2\yy$
to get
	\begin{equation} \label{e.eventually} f(\zz+\lambda_i \basis_i) - f(\zz)
= f(2\lambda_i \basis_i) - f(\lambda_i \basis_i) \qquad \text{for all } \zz \geq \lambda_i \basis_i, \end{equation}
which shows that $f-L$ is eventually periodic. Thus $f$ satisfies (i)--(iii).

Conversely, given an increasing $f = L+P$, define an equivalence relation on
$\N^k$ by $\yy \equiv \yy'$ if $f(\yy+\zz) - f(\yy) = f(\yy'+\zz)-f(\yy')$
for all $\zz \in \N^k$. If $L$ is linear and $P$ is eventually periodic, then
there are only finitely many equivalence classes: if $y_i \geq r_i +
\lambda_i$ then $\yy \equiv \yy - \lambda_i \basis_i$, so any $\yy \in \N^k$
is equivalent to some element of the cuboid $[0,\lambda_1+r_1] \times \cdots
\times [0,\lambda_k+r_k]$.

Now consider the abelian processor on the finite state space $\N^k / \equiv$
with $t_i(\xx) = \xx+ \basis_i$ and $o_i(\xx) = f(\xx+\basis_i) - f(\xx)$.
Note that $t_i$ and $o_i$ are well-defined. With initial state $\zero$, this
processor computes $f$.
\end{proof}

\subsection{Recurrent abelian processors}

Recall that an abelian processor is called \df{recurrent} if for any states
$q,q'\in Q$ there exists $\xx \in \N^k$ such that $q'=t_{\xx}(q)$.  Since we
assume that every state is accessible from the initial state $q^0$, this is
equivalent to the assertion that for every $q \in Q$ and $\yy \in \N^k$ there
exists $\zz \in \N^k$ such that $q = t_{\yy+\zz}(q)$.  Our next result
differs from Theorem~\ref{t.eventually} in only two words: \emph{recurrent}
has been added and \emph{eventually} has been removed!  As mentioned earlier,
we call a function satisfying (i)--(iii) below \df{ZILP}.

\begin{samepage}
\begin{theorem}
\label{t.recurrent} Let $k, \ell \geq 1$. A function $f: \N^k \to \N^\ell$
can be computed by a recurrent finite abelian processor if and only if $f$
satisfies all of the following.
	\begin{enumerate}[{\em (i)}]
	\item $f(\zero)=0$.
	\item $f$ is increasing.
	\item $f = L+P$ for a linear function $L$ and a periodic function $P$.
	\end{enumerate}
\end{theorem}
\end{samepage}

\begin{proof}
By Theorem~\ref{t.eventually}, $f$ satisfies (i) and (ii) and $f=L+P$ with
$L$ linear and $P$ \emph{eventually} periodic.  To prove (iii) we must show
that equation \eqref{e.eventually} holds for \emph{all} $\zz \in \N^k$.
By recurrence, for any $\yy \in \N^k$ and any $i \in A$
there exists $\yy' \geq \lambda_i \basis_i$ such that $t_{\yy'}(q)=t_{\yy}(q)$.
Now taking $\xx = \lambda_i \basis_i$ in Lemma~\ref{l.parallelogram}, the linear terms cancel, leaving
	\[ P(\yy + \lambda_i \basis_i) - P(\yy) = P(\yy' + \lambda_i \basis_i) - P(\yy'). \]
The right side vanishes since $P$ is eventually periodic. Since $\yy \in \N^k$ was arbitrary, $P$ is in fact periodic.

Conversely, given an increasing $f=L+P$, we define an abelian processor
$\Proc$ on state space $\N^k / \equiv$ as in the proof of
Theorem~\ref{t.eventually}. If $L$ is linear and $P$ is periodic, then for
each $i=1,\ldots,k$ we have $\yy \equiv \yy - \lambda_i \basis_i$ whenever
$y_i \geq \lambda_i$. Now given any $\xx, \yy \in \N^k$ we find $\xx' \equiv
\xx$ with $\xx' \geq \yy$, so $\Proc$ is recurrent.
\end{proof}

\subsection{Abelian networks}
\label{s.network}

An \df{abelian network} $\Net$ is a directed multigraph
$G=(V,E)$ along with specified pairwise disjoint sets
$I,O,T\subset E$ of \df{input}, \df{output} and \df{trash}
edges respectively. These edges are dangling: the input
edges have no tail, while the output and trash edges have
no head.  The trash edges are for discarding unwanted
letters.
Each node $v \in V$ is labeled with an abelian processor $\Proc_v$ whose
input alphabet equals the set of incoming edges to $v$ and whose output
alphabet is the set of outgoing edges from $v$.  In this paper, all abelian
networks are assumed finite: $G$ is a finite graph and each $\Proc_v$ is a
finite processor.

An abelian network operates as follows.  Its total state is given by the internal
states $(q_v)_{v\in V}$ of all its processors $\Proc_v$, together with a vector
$\xx=(x_e)_{e\in E}\in \N^E$ that indicates the number of letters sitting on
each edge, waiting to be processed. Initially, $\xx$ is supported on the set
of input edges $I$.  At each step, any non-output non-trash edge $e$ with $x_e>0$ is
chosen, and a letter is fed into the processor at its endnode $v$.  Thus,
$x_e$ is decreased by $1$, the state of $\Proc_v$ is updated from $q_v$ to
$t_e(q_v)$, and $\xx$ is increased by $o_e(q_v)$ (interpreted as a vector in
$\N^E$ supported on the outgoing edges from $v$).  Here $t$ and $o$ are the
maps associated to $\Proc_v$.  The sequence of choices of the edges $e$ at
successive steps is called a \df{legal execution}.  The execution is said to
\df{halt} if, after some finite number of steps, $\xx$ is supported on the
set of output and trash edges (so that there are no letters left to process).

The following facts are proved in \cite[Theorem~4.7]{BL15a}. Fixing the
initial internal states $\qq^\start=(q_v^\start)_{v\in V}$ and an input
vector $\xx\in\N^I$, if some execution halts then all legal executions halt.
In the latter case, the final states of the processors and the final output
vector do not depend on the choice of legal execution.  

Suppose for a given $\qq^\start$ that the network halts on {\em all} input vectors.
Then, since the final output vector depends only on the input vector, the abelian
network computes a function $\N^{I} \to \N^{O}$.  If a network $\Net$ and a
processor $\Proc$ compute the same function, then we say that $\Net$
\df{emulates} $\Proc$.
\begin{prop}
\label{backwards} If a finite abelian network
 halts on all inputs, then it emulates some finite abelian processor.
\end{prop}
\begin{proof}
We can regard the entire network as a processor, with its state given by the
vector of internal states $\qq=(q_v)_{v\in V}$.  Its transition and output
maps are determined by feeding in a single input letter, performing any legal
execution until it halts, and observing the resulting state and output
letters.  Feeding in two input letters and using (a special case of) the
insensitivity to execution order stated above, we see that the relations
\eqref{commute} hold, so the processor is abelian.
\end{proof}


The abelian networks that halt on all inputs are characterized in \cite[Theorem 5.6]{BL15b}: they are those for which a certain matrix called the
production matrix has Perron-Frobenius eigenvalue strictly less than $1$. An
abelian network is called \df{directed acyclic} if its graph $G$ has no
directed cycles; such a network trivially halts on all inputs.  This paper is
mostly concerned with directed acyclic networks, together with some networks
with certain limited types of feedback; all of them halt on all inputs.

\subsection{Recurrent abelian networks}

The next lemma follows from  \cite[Theorem~3.9]{BL15c}, but
we include a proof for the sake of completeness. A
processor is called \df{immutable} if it has just one
state, and \df{mutable} otherwise. Among the abelian logic
gates in Table~\ref{fig:gates}, splitters and adders are
immutable; topplers, delayers and presinks are mutable.

Recall (from Proposition~\ref{p.feedback} or Figure~\ref{fig:cycles}) that if feedback is permitted, then the transient $\Del$ gate can be emulated by a network of recurrent gates, namely a $2$-$\Topp$ and a $\Spl$.  The next result shows that without feedback, no transient processor can be emulated by a network of recurrent processors.

\begin{prop}
\label{l.acyclic}
A directed acyclic network $\Net$ of recurrent processors emulates a recurrent processor.
\end{prop}

\begin{proof}
We proceed by induction on the number $m$ of mutable processors in $\Net$. In
the case $m=0$, the network $\Net$ has only one state, so it is trivially
recurrent.

\old{
In the case $m=1$, the network consists of an immutable piece $\ProcI$
feeding into a recurrent processor $\Proc$ which feeds into an immutable
piece $\ProcJ$. Also $\ProcI$ can feed directly into $\ProcJ$. Writing
$(I_1,I_2)$ and $J$ for the functions computed by $\ProcI$ and $\ProcJ$
respectively, we have
	\[ f_{\Net,q}(\xx) = J( f_{\Proc,q}(I_1(\xx))+I_2(\xx)). \] Since
$\ProcI$ and $\ProcJ$ are immutable, $I_1, I_2, J$ are linear functions. By
Theorem~\ref{t.recurrent} the right side is the sum of a linear and a
periodic function since $f_{\Proc,q}$ is. }

For the inductive step, suppose $m \geq 1$. Since $\Net$ is
directed acyclic, it has a mutable processor $\Proc$ such
that no other mutable processor feeds into anything
upstream of $\Proc$. If $\Net$ has $k$ inputs, we can
regard the remainder $\Net - \Proc$ as a network with
$m{-}1$ mutable processors and $k{+}k'$ inputs, where $k'$
is the number of edges from $\Proc$ to $\Net-\Proc$. For
each state $q$ of $\Net$, the function $f = f_{\Net,q}$ has
the form
	\[ f(\xx) = g(\xx, h(\xx)) + j(\xx) \] where $g : \N^{k+k'} \to \N^\ell$
is the function computed by $\Net-\Proc$ in initial state $q$; and $h: \N^k
\to \N^{k'}$ and $j: \N^k \to \N^\ell$ are the functions sent by $\Proc$ in
initial state $q$ to $\Net-\Proc$ and the output of $\Net$, respectively.

By Theorem~\ref{t.recurrent} and the inductive hypothesis, each of $g,h,j$ is
the sum of a periodic and a linear function. Writing $g(\yy) = P(\yy) + \bb
\cdot \yy$ and $h(\xx) = Q(\xx) + \cc \cdot \xx$ we have
\[ f(\xx) = P(\xx, Q(\xx)+ \cc \cdot \xx) + \bb \cdot (\xx, Q(\xx) + \cc \cdot \xx) + j(\xx). \]
The first term is periodic and the second is a linear function plus a
periodic function. Since $q$ is arbitrary the proof is complete by
Theorem~\ref{t.recurrent}.
\end{proof}

\subsection{Varying the initial state}

We remark that the emulation claims of our main theorems can be strengthened
slightly, in the following sense.  Our definition of a processor $\Proc$
included a designated initial state $q^0$, but one may instead consider
starting $\Proc$ from any state $q\in Q$, and it may compute a different
function from each $q$.  All of these functions can be computed by the
\emph{same} network $\Net$, by varying the internal states of the gates in
$\Net$. To set up the network $\Net$ to compute the function $f_{\Proc,q}$,
we simply choose an input vector $\xx$ that causes $\Proc$ to transition from
$q^0$ to $q$, then feed $\xx$ to $\Net$ and observe the resulting gate
states.  In the recurrent case, this amounts to adjusting the priming of
topplers.  In the transient case, a ``used'' delayer can be replaced with a
wire, while a used presink becomes a trash edge.

\subsection{Splitter-adder networks}
\label{s.splitteradder}

In this section we show that splitter-adder networks
compute precisely the increasing linear functions. Using this, we
will see how Theorem~\ref{t.recurrent.intro} implies
\cref{t.functions}.

\begin{lemma}
\label{l.linear} Let $k, \ell \geq 1$. The function $f :
\N^k \to \N^\ell$ can be computed by a network of splitters
and adders if and only if $f(\xx)=L\xx$ for some
nonnegative integer $\ell \times k$ matrix $L$.
\end{lemma}

\begin{proof}
If a network of splitters and adders has a directed cycle, then it does not
halt on all inputs, and so does not ``compute a function'' according to our
definition.  If the network is directed acyclic then by
\cref{l.acyclic,t.recurrent} it computes a ZILP function.  Since the network
is immutable, the periodic part of any linear + periodic decomposition must
be zero. Conversely, consider a network of $k$ splitters $\ProcS_i$ and
$\ell$ adders $\ProcA_j$, with $L_{ji}$ edges from $\ProcS_i$ to $\ProcA_j$.
Feed each input $i$ into $\ProcS_i$, and feed each $\ProcA_j$ into output
$j$.
\end{proof}

\begin{proof}[Proof of \cref{t.functions}]
Given a function $F \in \mathcal{R}$, the function $\xx
\mapsto F(\xx)-F(\zero)$ can be computed by a finite,
directed acyclic network of splitters, adders and (possibly
primed) topplers.
By \cref{l.acyclic} any such network emulates a recurrent
finite abelian processor, so $F$ has the desired form by
Theorem~\ref{t.recurrent}.

Conversely if $F=L+P$ with $L$ linear and $P$ periodic,
then $F(\xx)-F(\zero)$ is computable by a finite directed
acyclic network of splitters, adders and topplers by
Theorem~\ref{t.recurrent.intro}. We induct on the number of
topplers to show that $F \in \mathcal{R}$. In the base case
there are no topplers, $\Net$ is a splitter-adder network,
so by Lemma~\ref{l.linear}, $F$ is an increasing linear function
of its inputs $x_1, \ldots, x_k$.

Assume now that at least one component of $\Net$ is a toppler. Since $\Net$
is directed acyclic, there is a toppler $\ProcT$ such that no other toppler
is downstream of $\ProcT$. Write $\ProcD$ for the portion of $\Net$
downstream of $\ProcT$, and $\ProcU = \Net - \ProcT - \ProcD$ for the
remainder of the network. Suppose $\ProcU$ sends outputs $r,s,\uu$
respectively to the output of $\Net$, to $\ProcT$, and to $\ProcD$; and that
the toppler $\ProcT$ sends output $t$ to $\ProcD$.

The downstream part $\ProcD$ consists of only splitters and adders, so it
computes a linear function 	
	\[ L(t,\uu) = at+ \bb \cdot \uu\] for some $a \in \N$ and $\bb \in \N^j$,
where $j$ is the number of edges from $\ProcU$ to $\ProcD$. The total output
of $\Net$ is
	\[ F(\xx) - F(\zero) = r + L(t,\uu) = r + a\floor{\frac{s}{\lambda}} +
\bb \cdot \uu. \] Each of $r,s,\uu$ is a function of the input $\xx =
(x_1,\ldots,x_k)$. By induction, $r$ and $s$ and each $u_i$ belongs to the
class $\mathcal{R}$, so $F$ does as well.
\end{proof}

\section{Basic reductions}
\label{s.basic}

In this section we describe some elementary network reductions.

\old{
\begin{definition}
Processor $\Proc$ \emph{emulates} processor $\Proc'$ in initial state $q'$ if
there exists an initial state $q$ such that
	\[ f_{\Proc,q} = f_{\Proc',q'}. \]
\end{definition}
}

\subsection{Multi-way splitters and adders}
An $n$-\df{splitter} computes the function $\N \to \N^n$ sending $x \mapsto
(x,\ldots,x)$. It is emulated by a directed binary tree of $n-1$ splitters
with the input node at the root and the $n$ output nodes at the leaves.
Similarly, an $n$-\df{adder} computes the function $\N^n \to \N$ given by
$(x_1,\ldots,x_n)\mapsto x_1+\cdots+x_n$.  It is emulated by a tree of $n-1$
adders.

\subsection{The power of feedback}

\begin{proof}[Proof of Proposition~\ref{p.feedback}]
To emulate an unprimed $\lambda$-toppler, let $r = \ceiling{\log_2 \lambda}$ and let
	\[ 2^r - \lambda =  \sum_{i=0}^{r-2} b_i 2^i, \qquad
b_i \in \{0,1\} \] be the binary representation of $2^r -
\lambda$. Consider $r$ $2$-topplers
$H_0,H_1,\ldots,H_{r-1}$ in series: the input node is
$H_0$, and each $H_i$ feeds into $H_{i+1}$ for $0\leq i <
r{-}1$.  For $i < r-1$ the $2$-toppler $H_i$ is primed with
$b_i$. The last $2$-toppler $H_{r-1}$ is unprimed, and
feeds into an $s$-splitter ($s = 1+\sum_{i=0}^{r-1} b_i$)
which feeds one letter each into the output node $o$ and
the nodes $H_i$ such that $b_i=1$. This network repeatedly
counts in binary from $2^r - \lambda$ to $2^r -1$, and it
sends output precisely when it transitions from $2^r - 1$
back to $2^r - \lambda$. Hence, it emulates a
$\lambda$-toppler.  See Figure~\ref{fig:cycles} for
examples.

We can emulate a $q$-primed $\lambda$-toppler using the same network, but with
different initial states for its $2$-topplers.  The required states are simply those that result from feeding $q$ input letters into the network described above.

A delayer is constructed by splitting the output of a
$2$-toppler and adding one branch back in as input to the
$2$-toppler (Figure~\ref{fig:cycles}).
\end{proof}

\subsection{Primed topplers}
The following shows that we can also do without primed topplers (at the
expense of using a transient gate: the presink).

\begin{lemma}\label{l.unpriming}
A primed $\lambda$-toppler can be emulated by a directed acyclic network
comprising an unprimed $\lambda$-toppler, adders, splitters, and a presink.
\end{lemma}

\begin{proof}
See Figure~\ref{unpriming}.  For $0\leq q<\lambda$ we have
$\lfloor(x+q)/\lambda\rfloor=\lfloor(x+q(x-1)^+)/\lambda\rfloor$, so we can
emulate a $q$-primed $\lambda$-toppler by splitting the input, feeding it
into a presink, and adding $q$ copies of the result into the original input
before sending it to an unprimed $\lambda$-toppler.
\end{proof}

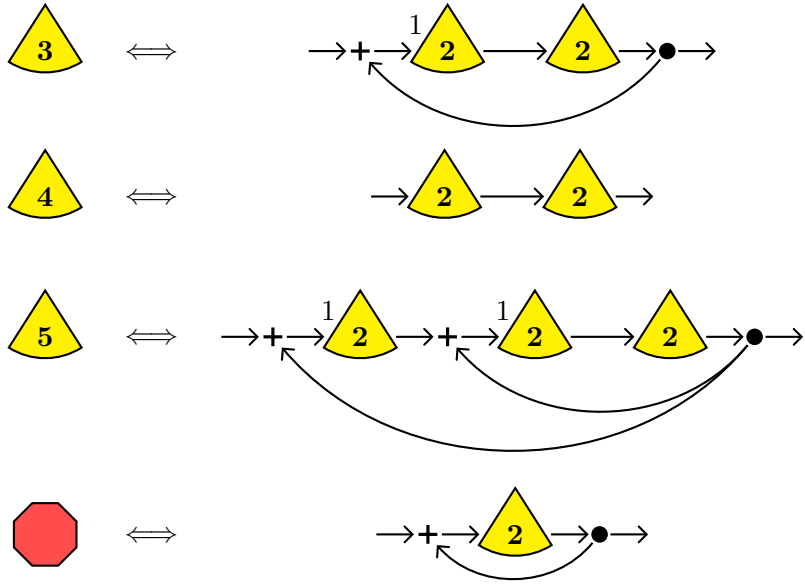
\begin{figure}
\centering
\setlength{\tabcolsep}{8pt}
\begin{tabular}{ccc}

\begin{tikzpicture}[abelian,baseline=-1mm]
  \node[toppler] (a) {$3$};
\end{tikzpicture}
&$\Longleftrightarrow$&
\begin{tikzpicture}[abelian,baseline=-1mm]
  \node[toppler,prime={$1$}] (a3) {$2$};
    \node[toppler] (b3) [right of=a3] {$2$};
  \node[splitter] (s) [right=.5cm of b3] {};
  \node[adder] (p) [left=5mm of a3] {};
  \draw[<-] (p.west) -- ++(-.5,0);
  \draw[->] (s.east) -- ++(.5,0);
  \draw
    (p) edge[->] (a3) (a3) edge[->] (b3) (b3) edge[->] (s)
    (s) edge[->,bend left=50] (p);
\end{tikzpicture}
\\

\begin{tikzpicture}[abelian,baseline=-1mm]
  \node[toppler] (a) {$4$};
\end{tikzpicture}
&$\Longleftrightarrow$&
\begin{tikzpicture}[abelian,baseline=-1mm]
  \node[toppler] (a) {$2$};
    \node[toppler] (b) [right of=a] {$2$};
  \draw[<-] (a.west) -- ++(-.5,0);
  \draw[->] (b.east) -- ++(.5,0);
  \draw
    (a) edge[->] (b);
\end{tikzpicture}
\\[1cm]

\begin{tikzpicture}[abelian,baseline=-1mm]
  \node[toppler] (a) {$5$};
\end{tikzpicture}
&$\Longleftrightarrow$&
\begin{tikzpicture}[abelian,baseline=-1mm]
  \node[toppler,prime={$1$}] (a) {$2$};
  \node[adder] (p2) [right=5mm of a] {};
  \node[toppler,prime={$1$}] (b) [right=5mm of p2] {$2$};
  \node[toppler] (c) [right of=b] {$2$};
  \node[splitter] (s) [right=.5cm of c] {};
  \node[adder] (p) [left=5mm of a] {};
  \draw[<-] (p.west) -- ++(-.5,0);
  \draw[->] (s.east) -- ++(.5,0);
  \draw[->] (p) -- (a);
  \draw[->] (a) -- (p2);
  \draw[->] (p2) -- (b);
  \draw[->] (b) -- (c);
  \draw[->] (c) -- (s);
    \draw
    (s) edge[->,bend left=50] (p)
    (s) edge[->,bend left=50] (p2);

\end{tikzpicture}
\\

\begin{tikzpicture}[abelian,baseline=-1mm]
  \node[delayer] (a) {};
\end{tikzpicture}
&$\Longleftrightarrow$&
\begin{tikzpicture}[abelian,baseline=-1mm]
  \node[toppler] (a) {$2$};
  \node[splitter] (s) [right=.5cm of a] {};
  \node[adder] (p) [left=5mm of a] {};
  \draw[<-] (p.west) -- ++(-.5,0);
  \draw[->] (s.east) -- ++(.5,0);
  \draw
    (p) edge[->] (a) (a) edge[->] (s)
    (s) edge[->,bend left=50] (p);
\end{tikzpicture}
\\
\end{tabular}

\caption{Emulating a $3$-toppler, $4$-toppler, $5$-toppler
and delayer by networks of $2$-topplers.}\label{fig:cycles}
\end{figure}

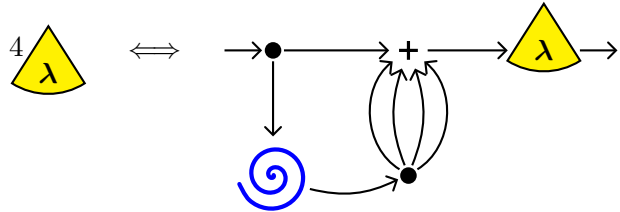
\begin{figure}
\centering
\setlength{\tabcolsep}{8pt}
\begin{tabular}{ccc}

\begin{tikzpicture}[abelian,baseline=2mm]
  \node[toppler,prime={$4$}] (a) {$\lambda$};
\end{tikzpicture}
&$\Longleftrightarrow$&
\begin{tikzpicture}[abelian,baseline=-1mm]
  \node[toppler] (a) {$\lambda$};
  \node[adder,outer sep=3pt] (p) [left of= a] {};
  \node[splitter] (s) [left of =p] {};
  \node[presink] (x) [below=10mm of s] {};
  \node[splitter] (xs) [right of=x] {};
  \draw[<-] (s.west) -- ++(-.5,0);
  \draw[->] (a.east) -- ++(.5,0);
  \draw (s) edge[->] (p) (p) edge[->] (a);
  \draw (s) edge[->] (x) (x) edge[->,bend right=20] (xs)
  (xs) edge[->,bend left=55] (p)
  (xs) edge[->,bend left=20] (p)
  (xs) edge[->,bend right=20] (p)
  (xs) edge[->,bend right=55] (p);
\end{tikzpicture}
\\
\end{tabular}
\caption{Emulating a primed toppler with an unprimed
toppler.}\label{unpriming}
\end{figure}

\subsection{Reduction to unary output}

Let $\Proc$ be an abelian processor that computes $f : \N^k \to \N^\ell$. If
$\ell=1$ then we say that $\Proc$ has \emph{unary output}. All of the logic
gates in Table~\ref{fig:gates} have unary output with the exception of the
splitter.  The next lemma shows that, for rather trivial reasons, it is
enough to emulate processors with unary output.

\begin{lemma}
\label{l.unaryoutput} Any abelian processor can be emulated by a directed
acyclic network of splitters and processors with unary output.
\end{lemma}

\begin{proof}
Let $\Proc$ compute $f = (f_1,\ldots,f_\ell) : \N^k \to \N^\ell$. By ignoring
all outputs of $\Proc$ except the $j$th, we obtain an abelian processor
$\Proc_j$ that computes $f_j$. Each $\Proc_j$ has unary output, and $\Proc$
is emulated by a network that sends each input into an $\ell$-splitter that
feeds into $\Proc_1, \ldots, \Proc_\ell$ (Figure~\ref{f.unaryoutput}).
\end{proof}
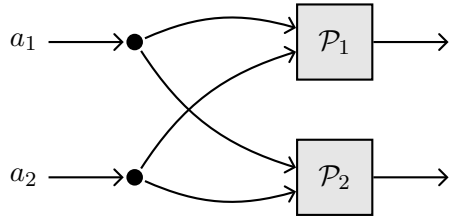
\begin{figure}
\centering
\begin{tikzpicture}[abelian]
  \node[block,minimum height=1cm,minimum width=1cm] (P1) {$\Proc_1$};
  \node[block,minimum height=1cm, minimum width=1cm,below of=P1] (P2) {$\Proc_2$};
  \node[splitter,left=2cm of P1] (a1) {};
  \node[splitter,left=2cm of P2] (a2) {};
  \node[left=1cm of a1] (aa1) {$a_1$};
  \node[left=1cm of a2] (aa2) {$a_2$};

  \draw[->]
  (aa1) edge (a1)
  (aa2) edge (a2)
  (a1) edge[bend left=20] (P1)
  (a1) edge[bend right=20] (P2)
  (a2) edge[bend left=20] (P1)
  (a2) edge[bend right=20] (P2);
  \draw
  (P1.east) edge[->] ++(1,0)
  (P2.east) edge[->] ++(1,0);
\end{tikzpicture}
\caption{Emulating a $2$-output abelian processor with two unary-output processors.}
\label{f.unaryoutput}
\end{figure}

In the subsequent proofs we can thus assume that the processors to be
emulated have unary output. By a \df{$k$-ary} processor we mean one with $k$
inputs.  A $1$-ary processor is sometimes called \df{unary}.

\section{The recurrent case}
\label{s.recurrent}

In this section we prove Theorem~\ref{t.recurrent.intro}.  By
\cref{l.unaryoutput} we may assume that the recurrent processor to be
emulated has unary output. We will proceed by induction on the number of
inputs.

\subsection{Unary case}

We start with the unary case (i.e.\ one input), which will
form the base of our induction.  An alternative would be to
start the induction with the trivial case of zero inputs,
but the simplicity of the unary case is illustrative.

By Theorem~\ref{t.recurrent}, a recurrent unary processor
computes an increasing function $f: \N \to \N$ of the form
$f(x) = cx + P(x)$, where $c \in \Q_{\geq 0}$ and $P: \N
\to \Q$ is periodic.

\begin{lemma}\label{unary-recurrent}
Let $\Proc$ be a recurrent unary processor that computes $f(x) = cx + P(x)$,
where $P$ is periodic of period $\lambda$. Then $\Proc$ can be emulated by a
network of adders, splitters and (suitably primed) $\lambda$-topplers.
\end{lemma}

\begin{proof}
Observe first that $c\lambda$ is an integer: since
$f(0)=0$, we have $P(\lambda) = P(0)=0$ thus $f(\lambda) =
c\lambda \in \N$. We construct a network of $c\lambda$
parallel $\lambda$-topplers as follows: the (unary) input
is split (by a $c\lambda$-splitter) into $c\lambda$
streams, each of which feeds into a separate
$\lambda$-toppler.  The outputs of the topplers are then
combined (by a $c\lambda$-adder) to a single output
(Figure~\ref{fig:topplers}).

After $m\lambda$ letters are input to this network, $m \in \N$, each toppler
will return to its original state having output $m$ letters, for a total
output of $m \times c\lambda = c(m\lambda)$; thus the network does compute
$cx+Q(x)$ where $Q$ has period $\lambda$ or some divisor of $\lambda$.  To
force $Q=P$ it suffices to choose the initial state $q$ in such a way that
the network's output for $x = 1,2,\dots,\lambda$ matches
$f(1),\dots,f(\lambda)$. This is easily done by setting $d_i =
f(x)-f(x{-}1)$, and for each $i$ with $1 \le i \le \lambda$, starting $d_i$
topplers in state $\lambda{-}i$.
\end{proof}

Figure~\ref{fig:topplers} illustrates the network constructed to compute the
function $f = \frac34 x + P(x)$ where $P$ has period 4 with $P(0)=0$,
$P(1)=\frac14$, $P(2)=\frac64$ and $P(3)=\frac34$.  The values of $f$ begin
$0,1,3,3,3,4,6,6,6,7,9,9,9,\dots$.  The ``I/O diagram" of $f$ is shown at the
top of the figure; dots represent input letters and bars are output letters;
the unfilled circle represents the initial state.

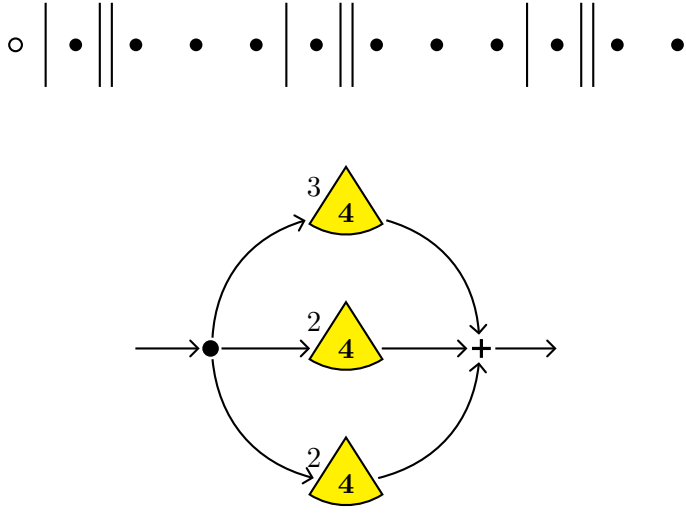
\begin{figure}
\centering
\begin{tikzpicture}[contour,scale=.8]
   \node[circ] at (0,0) {};
   \foreach \x in {1,...,11} {
   \node[dot] at (\x,0) {};  } ;
   \foreach \x in {1,1.9,2.1}
   \foreach \i in {0,4,8} { {
   \draw (\x+\i-.5,-.7) -- ++(0,1.4); } } ;
\end{tikzpicture}
\\[1cm]
\begin{tikzpicture}[abelian]
  \node[toppler,prime={$3$}] (a) {$4$};
  \node[toppler, below of=a,prime={$2$}] (b) {$4$};
  \node[toppler, below of=b,prime={$2$}] (c) {$4$};
  \node[splitter,left of=b] (i) {};
  \node[adder, right of=b] (o) {};
  \draw[->]
  (i) edge[bend left=35] (a) edge (b) edge[bend right=35] (c);
  \draw[<-]
  (o) edge[bend right=35] (a) edge (b) edge[bend left=35] (c);
  \draw (o) edge[->] ++(1,0);
  \draw (i) edge[<-] ++(-1,0);
\end{tikzpicture}
\caption{Emulating a unary processor with a network of primed topplers.}\label{fig:topplers}
\end{figure}

\subsection{Reduction to the meager case}
\label{s.meager}

A recurrent $k$-ary processor computes a function $f: \N^k
\to \N$ of the form $f(\xx)=\bb\cdot \xx + P(\xx)$ where
$P$ is periodic.

\begin{definition}
A recurrent processor is \df{nondegenerate} if $b_i \neq 0$
for all $i$.
\end{definition}

Note that if $b_i=0$ then $f(\xx)$ does not depend on the
coordinate $x_i$. In this case, by
Theorem~\ref{t.recurrent} there is a finite $(k{-}1)$-ary
recurrent processor that computes $f$.

Denote the lattice of periodicity of $P$ by $\Lambda
\subset \Z^k$.  Let $\lambda_i$ be the smallest positive
integer such that $\lambda_i \basis_i \in \Lambda$.  For
the purposes of the forthcoming induction, we focus on the
last coordinate.

\begin{definition} We say that a $k$-ary processor $\Proc$ is \df{meager} if
$f_\Proc(\lambda_k \basis_k)=1$.
\end{definition}

Note that if $\Proc$ is meager then for all $\xx \in \N^k$ we have
	 \[ f_\Proc(\xx+\lambda_k \basis_k) = f_\Proc(\xx) + 1. \] Next we
emulate a nondegenerate recurrent processor by a network of meager
processors.

\begin{lemma}
\label{l.meager} Let $\Proc$ be a nondegenerate recurrent $k$-ary processor
and let $m = f_\Proc(\lambda_k \basis_k) = \lambda_k b_k$.
 Then $\Proc$
can be emulated by a network of $m-1$ splitters, $m-1$ adders, and $m$ meager
recurrent $k$-ary processors.
\end{lemma}

\begin{proof}
For each $j=0,\ldots,m-1$ consider the function
\[ f_j(\xx) = \floor{ \frac{f(\xx)+j}{m} }. \]
We claim that $f_j=f_{\Proc_j}$ for some recurrent processor $\Proc_j$.  One
way to prove this is to use Theorem~\ref{t.recurrent}, checking from the
above formula that since $f$ is ZILP, $f_j$ is also ZILP.  Another route is
to note that $f_j$ is computed by a network in which the output of $\Proc$ is
fed into a $j$-primed $m$-toppler.  By \cref{l.acyclic}, $f_j$ is therefore
computed by some recurrent processor.  (Note however that this network itself
will not help us to emulate $\Proc$ using gates, since it contains $\Proc$!)
\cref{f.meager} illustrates an example of the reduction.

Now we use the meagerization identity \eqref{e.meagerization}:
	\[ f= \floor{\frac{f}{m}} + \floor{\frac{f+1}{m}} + \cdots +
\floor{\frac{f+m-1}{m}} \] Thus, $\Proc$ is emulated by an $m$-splitter that
feeds into $\Proc_0,\ldots,\Proc_{m-1}$, with the results fed into an
$m$-adder. It remains to check that each $\Proc_j$ is meager.
We have
\[ f_j(\xx + \lambda_k \basis_k)
= \floor{ \frac{f(\xx+\lambda_k \basis_k)+j}{m}} =
\floor{ \frac{f(\xx) + \lambda_k b_k+j}{m}} = f_j(\xx)+1. \qedhere \]
\end{proof}

\old{
The functions $f_j$ in the proof of Lemma~\ref{l.meager} satisfy for $i=2,\ldots,k$
	\[ f_j(\xx + m \lambda_i \basis_i)
	= \floor{ \frac{f(\xx) + m\lambda_i b_i+j}{m}} = f_j(\xx)+\lambda_i b_i  \]
so
	\[ f_j(\xx) = \bb' \cdot \xx + P_j(\xx) \]
with $P_j$ is periodic modulo $\lambda'_1 \Z \times \cdots \times \lambda'_k \Z$, where
$\lambda'_1 = \lambda_1$ and $\lambda'_i = m\lambda_i$ for $i=2,\ldots,k$ and
$\lambda'_1 b'_1 = 1$ and $\lambda'_i b'_i = \lambda_i b_i$ for $i=2,\ldots,k$.
}

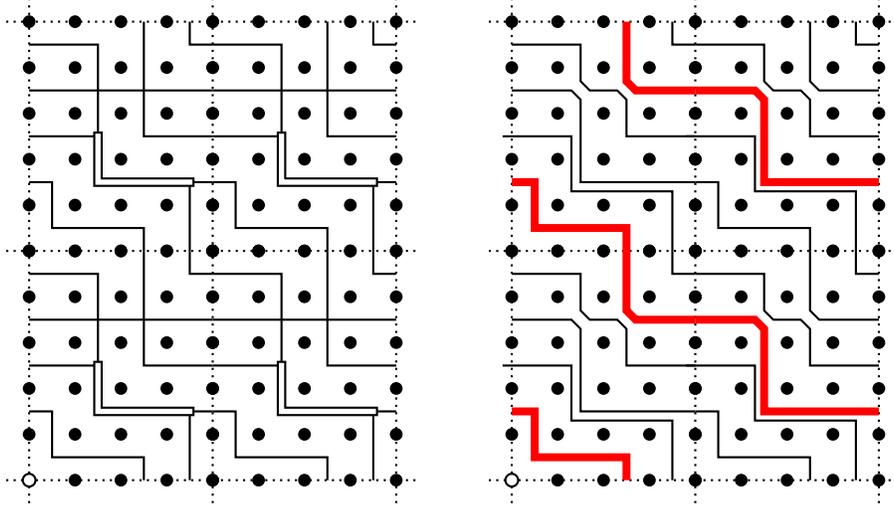
\begin{figure}
\centering
\begin{tikzpicture}[contour,scale=.61,cm={0,1,1,0,(0cm,0cm)}]
   \def\r{-> ++(1,0)}
   \def\hr{-> ++(.5,0)}
   \def\d{-> ++(0,-1)}
   \def\hd{-> ++(0,-.5)}

  \draw[xstep=5,ystep=4,dotted] (-.5,-.5) grid (10.5,8.5);
  \foreach \i in {0,5} {
  \foreach \j in {0,4} {
   \begin{scope}[shift={(\i,\j)}]
   \foreach \x in {0,...,5} {
   \foreach \y in {0,...,4} {
   \node[dot] at (\x,\y) {};  } } ;
   \draw (0,2.5) \hr\d\d\r\hd;
   \draw (0,3.5) \hr\r;
   \draw (1.5,4) \hd;
   \draw (2.5,1.5) \d\hd;
   \draw (2.5,1.5) \r\r\d\hd;
   \draw (2.5,4) \hd\d\r\r\hr;
   \draw (3.5,4) \hd\d\d\d\hd;
   \draw (4.5,4) \hd\hr;
   \draw[style=double, double distance=2pt,line cap=rect] (1.5,3.5) \d\d\r ;
   \end{scope};
  } };
     \node[circ] at (0,0){};
\end{tikzpicture}
\qquad
\begin{tikzpicture}[contour,scale=.61,cm={0,1,1,0,(0cm,0cm)}]
   \def\r{-> ++(1,0)}
   \def\hr{-> ++(.5,0)}
   \def\d{-> ++(0,-1)}
   \def\hd{-> ++(0,-.5)}
   \def\lr{-> ++(.8,0)}
   \def\gr{-> ++(1.2,0)}
   \def\ld{-> ++(0,-.8)}
   \def\gd{-> ++(0,-1.2)}
   \def\dd{-> ++(.2,-.2)}

  \draw[xstep=5,ystep=4,dotted] (-.5,-.5) grid (10.5,8.5);

  \foreach \i in {0,1} {
   \begin{scope}[shift={(\i*5,0)}]
   \foreach \x in {0,...,5} {
   \foreach \y in {0,...,4} {
   \node[dot] at (\x,\y) {};  } } ;
   \draw[hili] (0,2.5) \hr\d\d\r\hd;
   \draw (0,3.5) \hr\lr\d\gd\gr\d\hd;
   \draw (1.5,4) \hd\d\d\r\lr\dd\ld\hd;
   \draw (2.5,4) \hd\d\lr\dd -> ++(0,-.6) \dd\lr\d\hd;
   \draw[hili] (3.5,4) \hd\ld\dd\lr\hr;
   \draw (4.5,4) \hd\hr;
   \end{scope};
   };

  \foreach \i in {0,1} {
   \begin{scope}[shift={(5*\i,4)}]
   \foreach \x in {0,...,5} {
   \foreach \y in {0,...,4} {
   \node[dot] at (\x,\y) {};  } } ;
   \draw (0,2.5) \hr\d\d\r\hd;
   \draw (0,3.5) \hr\lr\d\gd\gr\d\hd;
   \draw[hili] (1.5,4) \hd\d\d\r\lr\dd\ld\hd;
   \draw (2.5,4) \hd\d\lr\dd -> ++(0,-.6) \dd\lr\d\hd;
   \draw (3.5,4) \hd\ld\dd\lr\hr;
   \draw (4.5,4) \hd\hr;
   \end{scope};
   };
     \node[circ] at (0,0){};
\end{tikzpicture}
 \caption{\emph{Left:} Example state diagram of a recurrent binary
processor $\Proc$ with $\lambda = (4,5)$ and
$b=(\frac12,\frac45)$.  (The function is the same as the
one in \cref{bar-graph} (left)).  A dot with coordinates
$\xx=(x_1,x_2)$ represents the state of the processor after
it has received input $\xx$. (The initial state $(0,0)$ is
an unfilled circle.) Each solid contour line between two
adjacent dots indicates that a letter is emitted when
making that transition. \emph{Right:} The highlighted
contours form the state diagram of the corresponding meager
processor $\Proc_3$, obtained by keeping every fourth
contour (starting from the first) of the left picture. The
vertical period is still $5$, although the horizontal
period has increased.} \label{f.meager}
\end{figure}


	
\subsection{Reducing the alphabet size}

Now we come to the main reduction.
\begin{lemma}
\label{l.mainreduction} Let $\Proc$ be a meager recurrent $k$-ary processor
satisfying $f_\Proc(\xx+\lambda_k \basis_k) = f_\Proc(\xx)+1$. Then $\Proc$
can be emulated by a network of a recurrent $(k{-}1)$-ary processor, a
$\lambda_k$-toppler, and an adder.
\end{lemma}

\begin{proof}
Let $\Proc$ compute $f$. By Theorem~\ref{t.recurrent}, $f$ is ZILP.  Its
representation as a linear plus a periodic function makes sense as a function
on all of $\Z^k$.  Now consider the increasing function
	\[  g(x_1,\ldots,x_{k-1}) = -c - \min\{x_k \in \Z \,:\, f(x_1,\ldots,x_k)
\geq 0 \}. \]
where $c = -\min \{x_k \in \Z \,:\, f(0,\ldots,0,x_k) \geq 0\}$. Note that
$g$ is an increasing function of $(x_1,\ldots,x_{k-1}) \in \N^{k-1}$, and
$g(\zero) = 0$.

If $\ld \in  \lambda_1 \Z \times \cdots \times \lambda_{k-1}
\Z\times\{0\} $, then
	\begin{align*} g(\xx+\ld) &= -c -\min \{x_k \in \Z \,:\,
f(x_1,\ldots,x_k)
+ \bb \cdot \ld \geq 0\} \\
		&= -c -\min \{x_k \in \Z \,:\, f( x_1,\ldots,x_{k-1},x_k + \lambda_k (\bb \cdot \ld)) \geq 0\} \\
		&= g(\xx) + \lambda_k (\bb \cdot \ld)
	\end{align*}
where the second equality holds because $\Proc$ is meager. Hence $g$ is ZILP.
Let $\ProcQ$ be the $(k-1)$-ary processor that computes $g$ (which exists by \cref{t.recurrent}).

Note that for any integer $j$ we have that $f(x_1,\ldots,x_k) \geq j$ if and
only if $f(\xx-j\lambda_k \basis_k) \geq 0$, which in turn happens if and
only if $g(x_1,\ldots,x_{k-1})+x_k +c \geq j\lambda_k$. Hence
	\[ f(x_1,\ldots,x_k) = \floor{
\frac{g(x_1,\ldots,x_{k-1})+x_k+c}{\lambda_k}}.
	\]
 The definition of $c$ gives that $0\leq c <\lambda_k$, since $f(\zero)=0$
 and $f(-\lambda_k \basis_k)=-1$.
So $\Proc$ is emulated by the network that feeds the last input letter $a_k$
into a $\lambda_k$-toppler $\ProcT$ primed with $c$, and $a_1,\ldots,a_{k-1}$
into $\ProcQ$ which feeds into $\ProcT$.
\end{proof}


\begin{figure}
\centering
\begin{tikzpicture}[abelian]
  \node[block,minimum height=2cm,minimum width=1cm] (Q) {$\ProcQ$};
  \draw[<-] (Q.130) -- ++(-1,0) node[left] (a1) {\makebox[1.3em]{$x_{1}$}};
  \draw[<-] (Q.230) -- ++(-1,0) node[left] (ak1) {\makebox[1.3em]{$x_{k-1}$}};
  \path (Q.180)  -- ++(-1,0) node[left] {\raisebox{1ex}{\makebox[1.4em]{$\vdots$}}};
  \node (ak) [below= 1cm of ak1] {\makebox[1em]{$x_{k}$}};
  \node[adder] (p) [right=3cm of ak1] {};
  \node[toppler] (b) [right=.5cm of p,prime=$c$] {$\lambda_k$};
  \draw (ak) edge[->, bend right=20] (p)
        (Q.east) edge[->, bend left=10] (p)
        (p) edge[->] (b)
        (b.east) edge[->] ++(1,0);
\end{tikzpicture}
\caption{Emulating a meager recurrent $k$-ary processor via a recurrent $(k{-}1)$-ary processor.}
\label{fig:induce}
\end{figure}
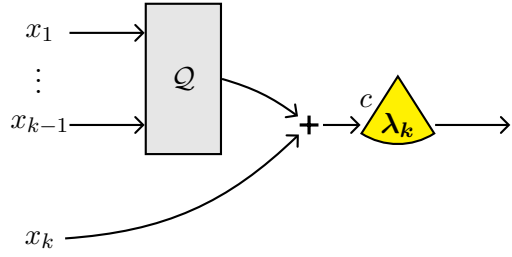

Now we can prove the main result in the recurrent case.
\begin{proof}[Proof of \cref{t.recurrent.intro}]
Let $\Proc$ be a recurrent abelian processor to be emulated.  By
\cref{l.unaryoutput} we can assume that it has unary output.  We proceed by
induction on the number of inputs $k$.  The base case $k=1$ is
\cref{l.unaryoutput}.  For $k>1$, we first use \cref{l.meager} to emulate the
processor by a network of meager $k$-ary processors.  Then we replace each of
these with a network of $(k-1)$-ary processors, by \cref{l.mainreduction},
and then apply the inductive hypothesis to each of these.
\end{proof}

\pagebreak[4]
\subsection{The number of gates}

How many gates do our networks use?  For simplicity, consider the case of a
recurrent $k$-ary abelian processor with $\lambda_i=2$ and $b_i=1/2$ for all
$i$.  It is not difficult to check that our construction uses $O(c^k)$ gates
as $k\to\infty$ for some $c$.  In fact, a counting argument shows that
exponential growth with $k$ is unavoidable, as follows.  Consider networks of
only adders, splitters, and $2$-topplers, but suppose that we allow feedback
(so that a $\lambda$-toppler can be replaced with $O(\log \lambda)$ gates, by
\cref{p.feedback}).  The number of networks with at most $n$ gates is at most
$n^{c' n}$ for some $c'$ (we choose the type of each gate, together with the
matching of inputs to outputs).  On the other hand, the number of different
ZILP functions $f$ that can be computed by a processor of the above-mentioned
form is at least $2^{\binom{k}{\lfloor k/2\rfloor}}$, since we may choose an
arbitrary value $f(\xx)\in\{0,1\}$ for each of the $\binom{k}{\lfloor
k/2\rfloor}$ elements $\xx$ of the middle layer $\{\xx\in\{0,1\}^k: \sum_i
x_i =\lfloor k/2\rfloor\}$ of the hypercube.  If all $k$-ary processors can
be emulated with at most $n$ gates then $n^{c'n}>2^{\binom{k}{\lfloor
k/2\rfloor}}$. It follows easily that some such processor requires at least
$C^k$ gates, for some fixed $C>1$.

If we consider the dependence on the quantities $\lambda_i$ and $b_i$ as well
as $k$, our construction apparently leaves more room for improvement
in terms of the number of gates, since repeated meagerization tends to
increase the periods $\lambda_i$.  One might also investigate whether there
is an interesting theory of $k$-ary functions that can be computed with only
\emph{polynomially} many gates as a function of $k$.

Our construction of networks emulating transient processors
(\cref{s.transient}) will be much less efficient than the recurrent case,
since the induction will rely on a ZILP function of a potentially large
number of arguments (\cref{p.blackbox}) that is emulated by appeal to
Theorem~\ref{t.recurrent.intro}.  It would be of interest to reduce the
number of gates here.

\section{The bounded case}
\label{s.bounded}

In this section we prove \cref{t.bounded.intro}.  Moreover,
we identify the class of functions computable without
topplers.

\begin{lemma}
\label{l.01} Let $f : \N^k \to \{0,1\}$ be increasing with $f(\zero)=0$.
There is a directed acyclic network of adders, splitters, presinks and
delayers that computes $f$.
\end{lemma}

\begin{proof}
Let $M$ be the set of $\mm \in \N^k$ that are minimal (in the coordinate
partial order) in $f^{-1}(1)$.  By Dickson's Lemma \cite{D13}, $M$ is finite;
and $f(\xx)=1$ if and only if $\xx \ge \mm$ for some $\mm \in M$. Thus
\[
f(\xx) = \bigvee_{\mm \in M} \bigwedge_{i \in A} \one[x_i \ge m_i].
\]
The function $\one[x_i \ge m_i]$ is computed by $m_i-1$ delayers in series
followed by a presink.  The minimum ($\wedge$) or maximum ($\vee$) of a pair
of boolean ($\{0,1\}$-valued) inputs is computed by adding the inputs and
feeding the result into a delayer or a presink respectively.  The minimum or
maximum of any finite set of boolean inputs is computed by repeated pairwise
operations. See \cref{fig:boolean}.  The lemma follows.
\end{proof}

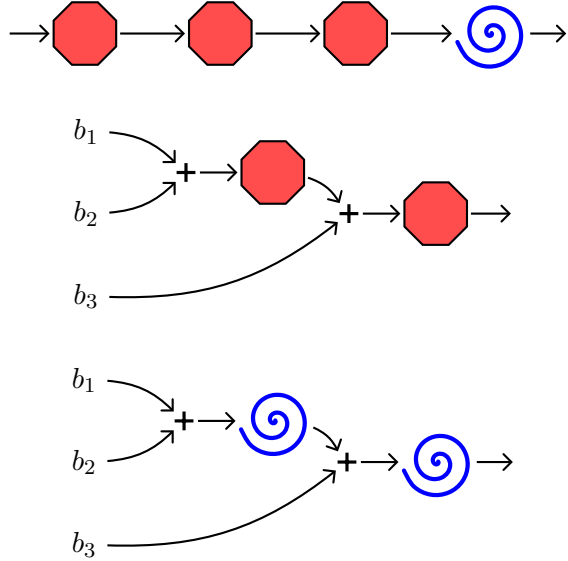
\begin{figure}
\centering
\begin{tikzpicture}[abelian]
\node[delayer] (d1) {};
\node[delayer,right of=d1] (d2) {};
\node[delayer,right of=d2] (d3) {};
\node[presink,right of=d3] (s) {};
\draw[->] (d1) -- (d2);
\draw[->] (d2) -- (d3);
\draw[->] (d3) -- (s);
\draw[->] (s) -- ++(1,0);
\draw[<-] (d1) -- ++(-1,0);
\end{tikzpicture}
\\[0.5cm]

\begin{tikzpicture}[abelian]
\node (b1) {$b_{1}$};
\node[below=1mm of b1] (x1) {};
\node[below=1mm of x1] (b2) {$b_{2}$};
\node[below=1mm of b2] (x2) {};
\node[below=1mm of x2] (b3) {$b_{3}$};
\node[adder,right=1cm of x1] (a1) {};
\node[adder,right=3cm of b2] (a2) {};
\node[delayer,right=.5cm of a1] (d1) {};
\node[delayer,right=.5cm of a2] (d2) {};
\draw (b1) edge[->,bend left=20] (a1);
\draw (b2) edge[->,bend right=20] (a1);
\draw (b3) edge[->,bend right=20] (a2);
\draw (a1) edge[->] (d1);
\draw (d1) edge[->,bend left=20] (a2);
\draw (a2) edge[->] (d2);
\draw (d2) edge[->] ++(1,0);
\end{tikzpicture}
\\[0.5cm]

\begin{tikzpicture}[abelian]
\node (b1) {$b_{1}$};
\node[below=1mm of b1] (x1){};
\node[below=1mm of x1] (b2) {$b_{2}$};
\node[below=1mm of b2] (x2){};
\node[below=1mm of x2] (b3) {$b_{3}$};
\node[adder,right=1cm of x1] (a1){};
\node[adder,right=3cm of b2] (a2){};
\node[presink,right=.5cm of a1] (d1){};
\node[presink,right=.5cm of a2] (d2){};
\draw (b1) edge[->,bend left=20] (a1);
\draw (b2) edge[->,bend right=20] (a1);
\draw (b3) edge[->,bend right=20] (a2);
\draw (a1) edge[->] (d1);
\draw (d1) edge[->,bend left=20] (a2);
\draw (a2) edge[->] (d2);
\draw (d2) edge[->] ++(1,0);
\end{tikzpicture}

\caption{Networks computing $\one[x \geq 4]$, and
$\min(b_1,b_2,b_3)$ and $\max(b_1,b_2,b_3)$ for boolean
inputs $b_i\in\{0,1\}$.}\label{fig:boolean}
\end{figure}

\begin{lemma}
\label{l.bounded}
 Suppose $f : \N^k \to \N$ is increasing and bounded
with $f(\zero)=0$. Then there is a directed acyclic network of adders,
splitters, presinks and delayers that computes $f$.
\end{lemma}

\begin{proof}
By Lemma~\ref{l.01}, for each $j \in \N$ there is a  network of the desired
type that computes $\xx \mapsto \ind [ f(\xx) > j ]$. If $f$ is bounded by
$J$, then $f(\xx) = \sum_{j=0}^{J-1} \ind [f(\xx)>j]$, so we add the outputs
of these $J$ networks.
\end{proof}

\begin{proof}[Proof of \cref{t.bounded.intro}]
Let $\Proc$ be a bounded abelian processor.  By \cref{l.unaryoutput} we can
assume that it has unary output.  The function that it computes is increasing
and bounded, and maps $\zero$ to $0$.  Therefore, apply \cref{l.bounded}.
\end{proof}

What is the class of all functions computable by a network of adders,
splitters presinks and delayers? Let us call a function $P : \N^k \to
\N^\ell$ \textbf{eventually constant} if it is eventually periodic with all
periods $1$; that is, there exist $r_1, \ldots, r_k \in \N$ such that $P(\xx)
= P(\xx+\basis_i)$ whenever $\xx_i \geq r_i$.
(Note the relatively weak
meaning of this term: Such a function may admit multiple limits as some
arguments tend to infinity while the others are held constant,
as in our remark following the definition of \textbf{eventually periodic}.)

\begin{theorem}
\label{t.immutable} Let $k \geq 1$. A function $f : \N^k \to \N$ can be
computed by a finite, directed acyclic network of adders, splitters, presinks
and delayers if and only if it satisfies all of the following.
	\begin{enumerate}[{\em (i)}]
	\item $f(\zero)=0$.
	\item $f$ is increasing.
	\item $f = L+P$ for a linear function $L$ and an eventually constant function $P$.
	\end{enumerate}
\end{theorem}

\begin{proof}
Let $\Net$ be such a network, and let it compute $f$. Since
adders and splitters are immutable, and presinks and
delayers become immutable after receiving one input, the
internal state of $\Net$ can change only a bounded number
of times. In fact, for each $i=1,\ldots,k$ we have $t_i^{r}
= t_i^{r+1}$ where $r$ is the total number of presinks and
delayers in $\Net$.  Letting $b_i := f((r+1)\basis_i) -
f(r\basis_i)$, it follows from Lemma~\ref{l.parallelogram}
that
	\begin{equation} \label{e.eventuallylinear} f(\xx+\basis_i) - f(\xx) = b_i \end{equation}
whenever $x_i \geq r$. Note that $b_i \in \N$. Letting
	$ P(\xx) := f(\xx) - \bb \cdot \xx $,
we find that $P(\xx+\basis_i) = P(\xx)$ whenever $x_i \geq r$, so $P$ is eventually constant.

Conversely, suppose that $f$ satisfies (i)-(iii). Write $f=L+P$ for a linear
function $L(\xx) = \bb \cdot \xx$ with $\bb \in \N^k$, and an eventually
constant function $P$. Then there exist $r_1, \ldots, r_k$ such that
\eqref{e.eventuallylinear} holds for all $i=1,\ldots,k$ and all $\xx \in
\N^k$ such that $x_i \geq r_i$.  In particular, the function
	\[ g(\xx) := f(\xx) - \sum_{i=1}^k b_i (x_i -r_i)^+ \] is ZILP and
bounded. By Theorem~\ref{t.bounded.intro} there is a network $\Net$ of
adders, splitters, presinks and delayers that computes $g$. To compute $f$,
feed each input $x_i$ into a splitter which feeds into $\Net$ and into an
$r_i$-delayer followed by a $b_i$-splitter.
\end{proof}

\section{The general case}
\label{s.transient}

In this section we prove \cref{t.eventually.intro}.  As in
the recurrent case, the proof will be by induction on the
number of inputs, $k$, of the abelian processor. We
identify $\N^k$ with $\N^{k-1} \times \N$, and write
$(\yy,z) = \yy + z\basis_k$.  Meagerization will again play
a crucial role. A major new ingredient is ``interleaving of
layers''.

\subsection{The unary case}

As before, we first prove the case of unary input, although
an alternative would be to start the induction with the
trivial zero-input processor.

\begin{lemma}\label{unary-trans}
Any abelian processor with unary input and output can be
emulated by a directed acyclic network of adders,
splitters, topplers, presinks and delayers.
\end{lemma}

\begin{proof}
Let the processor $\Proc$ compute $F:\N\to\N$.  Since $F$
is ZILEP, it is linear plus periodic when the argument is
sufficiently large; thus, there exists $R\in\N$ such that
the function $G$ given by
$$G(x):=F(x+R)-F(R),\qquad x\in\N$$
is ZILP. We have
$$F(x)=G\bigl((x-R)^+\bigr)+\sum_{i=0}^{R-1} \bigl[F(i+1)-F(i)\bigr]\,\ind[x> i]$$
as is easily checked by considering two cases: when $x\leq R$ the first term
vanishes and the second telescopes; when $x\geq R$, the second term is $F(R)$
and we use the definition of $G$.

By \cref{unary-recurrent}, the function $G$ can be computed  by a network of
adders, splitters and topplers. Now to compute $F$, we feed the input $x$
into $R$ delayers in series.  For each $0\leq i<R$, the output after $i$ of
them is also split off and fed to a delayer, to give $\ind[x>i]$ (as in the
proof of \cref{l.bounded}); this is split into $F(i+1)-F(i)$ copies, while
the output $(x-R)^+$ of the last delayer is fed into a network emulating
$\ProcG$, and all the results are added. See \cref{f-unary-trans} for an
example.
\end{proof}

\begin{figure}
\centering
\begin{tikzpicture}[abelian]
  \node[splitter] (s1) {};
  \node[delayer,right of=s1] (d1) {};
  \node[splitter,right of=d1] (s2){};
  \node[delayer,right of=s2] (d2){};
  \node[right=5mm of d2] (dd) {$\cdots$};
  \node[splitter,right=5mm of dd] (sr){};
  \node[delayer,right of=sr] (dr){};
  \node[presink,below of=s1] (p1){};
  \node[presink,below of=s2] (p2){};
  \node[presink,below of=sr] (pr){};
  \node[block,below of=dr,minimum width=1cm,minimum height=1cm] (g){$\ProcG$};
  \node[adder,below of= g](a){};
  \node[splitter,below of=p2] (ss2){};
  \draw[<-] (s1)--++(-1,0);
  \draw[->] (s1)--(d1);
  \draw[->] (d1)--(s2);
  \draw[->] (s2)--(d2);
  \draw[->] (sr)--(dr);
  \draw[->] (s1)--(p1);
  \draw[->] (s2)--(p2);
  \draw[->] (sr)--(pr);
  \draw[->] (dr)--(g);
  \draw[->] (g)--(a);
  \draw[->] (pr)--(a);
  \draw[->] (p2)--(ss2);
  \draw[->] (ss2) edge[bend right=15] (a);
  \draw[->] (ss2) edge[bend left=15] (a);
  \draw[->] (ss2) edge (a);
  \draw[->] (a)--++(1,0);
\end{tikzpicture}
\caption{Emulating a transient unary processor. (In this example, the difference
$F(i+1)-F(i)$ takes values $0,3,\ldots,1$ for $i=0,1,\ldots,R-1$).
}\label{f-unary-trans}
\end{figure}
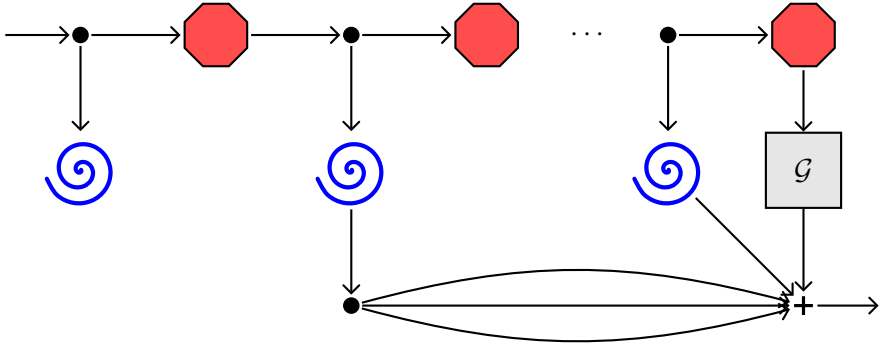

\subsection{Two-layer functions}

We now proceed with a simple case of the inductive step, which
provides a prototype for the main argument, and which will also be used as a
step in the main argument.

\begin{samepage}
\begin{lemma}\label{baby-case}
Let $\Proc$ be a $k$-ary abelian processor that computes a function $F$, and
suppose that
\begin{equation}\label{constant-after-1}
F(\yy,z)=F(\yy,z') \quad\text{if }z,z'\geq 1.
\end{equation}
Then $\Proc$ can be emulated by a network of topplers, presinks, and
$(k-1)$-ary abelian processors.
\end{lemma}
\end{samepage}

\begin{proof}
Define
$$W:=\sup_{\yy\in \N^{k-1}} F(\yy,1)-F(\yy,0),$$
and note that $W<\infty$, because the difference inside the supremum is an
eventually periodic function of $\yy$, and is thus bounded.  If $W=0$, then
$F$ is constant in $z$, and therefore $\Proc$ can be emulated by a single
$(k-1)$-ary processor.

Suppose $W\geq 1$.  We reduce to the case $W=1$ by the meagerization identity
\eqref{e.meagerization}.  Specifically, we express $F$ as $\sum_{i=0}^{W-1}
F_i$, where $F_i:=\lfloor(F+i)/W\rfloor$.  Each $F_i$ is ZILEP (this can be
checked directly, or by \cref{backwards}, since $F_i$ is computed by feeding
the output of $F$ into a toppler). Each $F_i$ satisfies the condition
\eqref{constant-after-1}, but now has $F_i(\yy,1)-F_i(\yy,0)\leq 1$ for all
$\yy$, as promised.  If we can find a network to compute each $F_i$ then the
results can be fed to an adder to compute $F$.

Now we assume that $W=1$.  Define the two $(k-1)$-ary functions (``layers''):
\begin{align*}
f_0(\yy)&:=F(\yy,0),\\
f_1(\yy)&:=F(\yy,1)-u,\qquad \text{where }u:=F(\zero,1).
\end{align*}
Each of $f_0,f_1$ is ZILEP.  Therefore, by \cref{t.eventually}, they are computed by suitable
$(k-1)$-ary abelian processors $\Proc_0,\Proc_1$.  Note that since $W=1$, we
have $u\in\{0,1\}$. We now claim that
\begin{equation}\label{2-layer}
F(\yy,z)=\biggl\lfloor
\frac{f_0(\yy)+f_1(\yy)+u+\ind[z>0]}{2}\biggr\rfloor.
\end{equation}
Once this is proved, the lemma follows: we split $\yy$ and feed it to both
$\Proc_0$ and $\Proc_1$, while feeding $z$ into a presink.  The three outputs
are added and fed into a primed $2$-toppler in initial state $u$.  See
\cref{f.2-layer}.
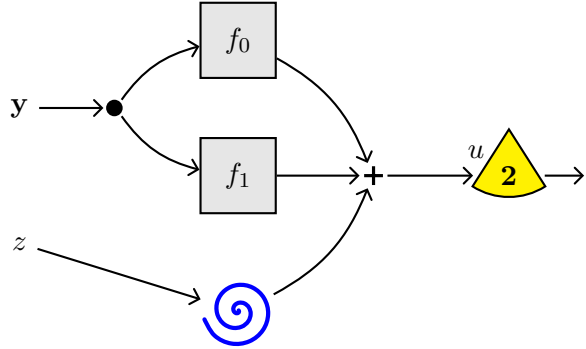
\begin{figure}
\centering
\begin{tikzpicture}[abelian]
  \node[block,minimum height=1cm,minimum width=1cm] (p0) {${f_0}$};
  \node[below of=p0,block,minimum height=1cm,minimum width=1cm] (p1) {${f_1}$};
  \node[below of=p1,presink] (o){};
  \path (p0) -- (p1) node[pos=.5,left=1.5cm,splitter] (s){};
  \node[right of=p1,adder] (a){};
  \node[right of=a,toppler,prime=$u$] (t){$2$};
  \draw[<-] (s) -- ++(-1,0) node[left] (y){$\yy$};
  \node[below of=y] (z){$z$};
  \draw[<-] (o) -- (z);
  \draw[->] (s) to[bend left=20] (p0);
  \draw[->] (s) to[bend right=20] (p1);
  \draw[->] (p0) to[bend left=20] (a);
  \draw[->] (o) to[bend right=20] (a);
  \draw[->] (p1) to (a);
  \draw[->] (a) to (t);
  \draw[->] (t) -- ++(1,0);
\end{tikzpicture}
\caption{Inductive step for emulating a two-layer function.
 (Here the solid disk represents $k-1$ parallel splitters that
 split each of the $k-1$ entries of the vector $\yy$ into two.)
}\label{f.2-layer}
\end{figure}

It suffices to check \eqref{2-layer} for $z=0$ and $z=1$, since both sides
are constant in $z\geq 1$.  Write $\Delta(\yy)=F(\yy,1)-F(\yy,0)$, so that
$\Delta(\yy)\in\{0,1\}$ for each $\yy$.  For $z=0$, the right side of
\eqref{2-layer} is
$$\biggl\lfloor
\frac{2F(\yy,0)+\Delta(\yy)}{2}
\bigg\rfloor
=F(\yy,0).
$$
On the other hand, for $z=1$ we obtain
$$\biggl\lfloor
\frac{2F(\yy,1)+(1-\Delta(\yy))}{2}
\bigg\rfloor
=F(\yy,1),
$$
as required.
\end{proof}

\subsection{A pseudo-minimum and interleaving}

The proof of \cref{t.eventually.intro} follows similar lines to the
proof above, but is considerably more intricate.  Again we will start by
using the meagerization identity to reduce to a simpler case.  The last step
of the above proof can be interpreted as relying on the fact that $\lfloor
(a+b)/2\rfloor=\min\{a,b\}$ if $a$ and $b$ are integers with $|a-b|\leq 1$.
We need a generalization of this fact involving the minimum of $n$ arguments.
The minimum function $(x_1,\ldots,x_n) \mapsto \min\{x_1,\ldots,x_n\}$ itself
is increasing but only \emph{piecewise} linear. Since it has unbounded
difference with any linear function, it cannot be expressed as the sum of a
linear and an eventually periodic function, and thus cannot be computed by a
finite abelian processor.  The next proposition states, however, that there
exists a ZILP function that agrees with $\min$ near the diagonal.  For the
proof, it will be convenient to extend the domain of the function from $\N^n$
to $\Z^n$. \cref{t.recurrent.intro} implies that the restriction of  such a
function to $\N^n$ can be computed by a recurrent abelian network of gates.

\begin{samepage}
\begin{prop}[Pseudo-minimum]
\label{p.blackbox} Fix $n\geq 1$.  There exists an increasing
function $M:\Z^n\to\Z$ with the following properties:
\begin{enumerate}[{\rm (i)}]
\item $M(\xx+n^2 \basis_j)= M(\xx)+n$, for all
    $\xx\in\Z^n$ and $1\leq j\leq n$; \label{cond1}
\item if $\xx\in\Z^n$ is such that $\max_j x_j-\min_j
    x_j\leq n-1$ then $M(\xx)=\min_j x_j$. \label{cond2}
\end{enumerate}
\end{prop}
\end{samepage}

The case $n=1$ of the above result is trivial, since we can take $M$ to be
the identity.  When $n=2$ we can take $M(\xx)=\lfloor(x_1+x_2)/2\rfloor$
(which satisfies the stronger periodicity condition $M(\xx+2 \basis_j)=
M(\xx)+1$ than (i)). The result is much less obvious for $n\geq 3$. Our proof
is essentially by brute force.  Our $M$ will in addition be symmetric in the
coordinates.
The period $n^2$ appearing in (i) is relatively unimportant. We do not know
whether it can be reduced to order $n$.  Any period suffices for our
application.

\begin{proof}[Proof of Proposition~\ref{p.blackbox}]
We start by defining a function $\pre$ that satisfies
the given conditions but is not defined everywhere.  Then
we will fill in the missing values. Let
$$\K:=\bigl\{x\in\Z^n: \textstyle\min_j x_j=0,\; \max_j x_j\leq n{-}1\bigr\}
=[0,n{-}1]^n\setminus [1,n{-}1]^n.$$ (This is the set on which
(ii) requires $M$ to be $0$.) Write
$\11=(1,\ldots,1)\in\Z^n$. Let $\pre:\Z^n\to\Z\cup\{\un\}$
be given by
\begin{equation}\label{partial-f}
\pre(\xx)=\begin{cases}
n\sum_ju_j +s &\text{if } \xx\in
\K+n^2 \uu + s\11
 \quad\text{for some }\uu\in\Z^n,\; s\in\Z,\\
 \un&\text{otherwise.}
\end{cases}
\end{equation}
Here the symbol $\un$ means ``undefined''.  See \cref{mhat} for an
illustration.
\begin{figure}
\centering
\begin{tikzpicture}[scale=.55]
  \clip (-.1,-.1) rectangle (7.1,7.1);
  \fill[green] (0,0)--(2,0)--(2,1)--(1,1)--(1,2)--(0,2)--cycle;
  \draw (0,0) grid (7,7);
  \foreach \x in {0,...,6} {
  \node at (\x+.5,\x+.5) {\x};
  \node at (\x+1.5,\x+.5) {\x};
  \node at (\x+.5,\x+1.5) {\x};
  \node at (\x+2.5,\x-1.5) {\x};
  \node at (\x+3.5,\x-1.5) {\x};
  \node at (\x+2.5,\x-.5) {\x};
  \node at (\x-1.5,\x+2.5) {\x};
  \node at (\x-1.5,\x+3.5) {\x};
  \node at (\x-.5,\x+2.5) {\x};
  }
\end{tikzpicture}
\caption{Part of the function $\pre$ when $n=2$.
The origin is at the bottom left, and the region $\K$ is shaded.}
\label{mhat}
\end{figure}

We first check that the above definition is
self-consistent.  Suppose that $\xx\in \K+n^2\uu+s\11$ and
$\xx\in \K+ n^2\vv+t\11$; we need to check that the
assigned values agree. First suppose that $\uu{-}\vv$ does
not have all coordinates equal.  Then
$$\bigl\|(n^2\uu+s\11)-(n^2\vv+t\11)\bigr\|_\infty=
\bigl\|n^2(\uu{-}\vv)+(s{-}t)\11\bigr\|_\infty \geq \frac
{n^2}{2},$$ since two coordinates of $n^2(\uu{-}\vv)$ differ
by at least $n^2$, and the same quantity $s{-}t$ is added to
each.  This gives a contradiction since $\K$ has
$\|\cdot\|_\infty$-diameter $n{-}1<n^2/2$.  Therefore,
$\uu{-}\vv$ has all coordinates equal, i.e.\ $\uu{-}\vv=w\11$
for some $w\in\Z$.  Since $\K+a\11$ and $\K+b\11$ are
disjoint for $a\neq b$, we must have
$n^2\uu+s\11=n^2\vv+t\11$, so $n^2 w=t{-}s$.  But then the
two values assigned to $\pre(\xx)$ by \eqref{partial-f} are
$n\sum_j u_j+s$ and $n(\sum_j u_j - nw)+t$, which are
equal.

Next observe that $\pre$ satisfies an analogue of (i).
Specifically,
\begin{equation}\label{f-per}
  \pre(\xx)\neq \un \quad\text{implies}\quad \pre(\xx+n^2\vv)=\pre(\xx)+n\sum_jv_j.
\end{equation}
This is immediate from \eqref{partial-f}.  Note also that
$\pre$ satisfies (ii), i.e.\
\begin{equation}\label{f-diag}
\max_j x_j-\min_j x_j\leq n{-}1 \quad\text{implies}\quad
\pre(\xx)=\min_j x_j,
\end{equation}
since the assumption is equivalent to $\xx\in \K+(\min_j
x_j)\11$.

The key claim is that $\pre$ is increasing where it is
defined:
\begin{equation}\label{f-mono}
\xx\leq\yy\text{ and } \pre(\xx)\neq\un\neq \pre(\yy)
\quad\text{imply}\quad
\pre(\xx)\leq \pre(\yy).
\end{equation}
  To prove this, suppose that $\xx\in
\K+n^2\uu+s\11$ and $\yy\in \K+ n^2\vv+t\11$ satisfy
$\xx\leq \yy$.  If $\uu{-}\vv$ has all coordinates equal then
we again write $\uu{-}\vv=w\11$, so $\yy\in
\K+n^2\uu+(t-n^2w)\11$.  For $a<b$, no element of $\K+a\11$
is $\geq$ any element of $\K+b\11$.  Therefore $\xx\leq\yy$
implies $s\leq t-n^2 w$, which yields $\pre(\xx)\leq
\pre(\yy)$ in this case. Now suppose $\ww:=\uu{-}\vv$ does
not have all coordinates equal, and write
$\overline{w}=n^{-1} \sum_j w_j$.  Since $\00\leq \K\leq
(n{-}1)\11\leq n\11$,
$$n^2\uu+s\11 \leq \xx\leq\yy \leq n^2\vv+t\11+n\11,$$
which gives $n^2\ww \leq (t-s+n)\11$. Suppose for a
contradiction that $\pre(\xx)>\pre(\yy)$, which is to say
$n\sum_j u_j+s>n\sum_j v_j+t$, i.e.\ $t-s<n^2\overline{w}$.
Combined with the previous inequality, this gives
$n^2\ww<(n+n^2\overline{w})\11$ (where $<$ denotes strict
inequality in all coordinates). That is
$$w_j-\overline{w}<\frac{1}{n},\qquad 1\leq j\leq n$$
which is impossible by the assumption on $\ww$.  Thus \eqref{f-mono} is proved.

Now we fill in the gaps: define $M$ by
$$M(\xx):=
 \sup\bigl\{\pre(\zz):\zz\leq \xx\text{ and }\pre(\zz)\neq\un\bigr\},
$$
where the supremum is $-\infty$ if the set is empty and
$+\infty$ if it is unbounded above.  (But these
possibilities will in fact be ruled out below).

If $\xx\leq\yy$ then the set in the definition of
$M(\xx)$ is contained in that for $M(\yy)$.  So
$M$ is increasing. If $\pre(\xx)\neq\un$ then taking
$\zz=\xx$ and using \eqref{f-mono} gives
$M(\xx)=\pre(\xx)$. In particular $M$ satisfies (ii)
by \eqref{f-diag}.  It is easily seen that for every $\xx$
there exist $\uu\leq\xx\leq\vv$ such that
$\pre(\uu)\neq\un\neq \pre(\vv)$.  Therefore monotonicity
of $M$ shows that $M(\xx)$ is finite. Finally, the
definition of $M$ and \eqref{f-per} immediately imply
that $F$ satisfies the same equality as $\pre$ in
\eqref{f-per} (now for all $\xx$), which is (i).
\end{proof}

The above result will be applied as follows.

\begin{lemma}
\moniker{Interleaving} \label{l.stronginterleaving} Fix $n,k \geq 1$. Let $F : \N^{k-1} \times \N
\to \N$ be an increasing function satisfying
	\begin{equation} \label{e.benign} F(\yy,z) \leq F(\yy,z{+}1) \leq F(\yy,z)+1. \end{equation}
for all $\yy \in \N^{k-1}$ and all $z \in \N$.  Then
	\[ F(\yy,z) = M \bigl( F(\yy,z_0), \ldots, F(\yy,z_{n-1}) \bigr) \]
where $M$ is the function from \cref{p.blackbox}, and
		\[ z_i =z_i(z):= n \floor{\frac{z{+}n{-}i{-}1}{n}} + i. \]
\end{lemma}

\begin{figure}
\centering
\begin{tikzpicture}[thick,scale=.45]
  \draw[magenta] (0,0)--(1,0)--++(0,4)--++(4,0)--++(0,4)--++(4,0) node[right](z0){$z_0$};
  \draw[orange] (0,1)--(2,1)--++(0,4)--++(4,0)--++(0,4)--++(3,0) node[right](z1){$z_1$};
  \draw[green!70!black] (0,2)--(3,2)--++(0,4)--++(4,0)--++(0,4)--++(2,0) node[right](z2){$z_2$};
  \draw[blue] (0,3)--(4,3)--++(0,4)--++(4,0)--++(0,4)--++(1,0) node[right](z3){$z_3$};
  \draw[black,->] (.5,-.2)--(10.5,-.2) node[right](z){$z$};
  \draw[black,<-] (.5,11.5)--(.5,-.2);
\end{tikzpicture}
\caption{The interleaving functions $z_i(z)$ for $n=4$.}\label{fig:interleave}
\end{figure}
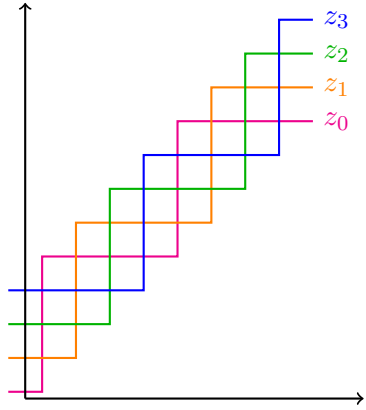

The idea is that the function $(\yy,z)\mapsto
F(\yy,z_i(z))$ appearing in \cref{l.stronginterleaving}
picks out every $n$th layer of $F$, starting from the $i$th
(with each such layer repeated $n$ times after an
appropriate initial offset). \cref{fig:interleave}
illustrates the functions $z_i$. The lemma says that we can
recover $F$ from these functions by ``interleaving'' their
layers, thus reducing the computation of $F$ to potentially
simpler functions. Note that the functions $(\yy,z)\mapsto
F(\yy,z_i(z))$ do not necessarily map $\zero$ to $0$ (even
if $F$ does), and so cannot themselves be computed by
abelian processors. We will address this issue with
appropriate adjustments (akin to the use of the quantity
$u$ in the proof of \cref{baby-case}) when we apply the
lemma in the proof of \cref{t.eventually.intro}.

\begin{proof}[Proof of \cref{l.stronginterleaving}]
As $i$ ranges from $0$ to $n{-}1$, note that $z_i$ takes on each of the
values $z,z{+}1,\ldots,z{+}n{-}1$ exactly once. Thus, the increasing
rearrangement of $(F(\yy,z_i))_{i=0}^{n-1}$ is $(F(\yy,z{+}j))_{j=0}^{n-1}$.
By \eqref{e.benign} it follows that
	\[ M((F(\yy,z_i))_{i=0}^{n-1}) = \min (F(\yy,z_i))_{i=0}^{n-1} = F(\yy,z). \qedhere \]
\end{proof}

\subsection{Proof of main result}

\begin{proof}[Proof of \cref{t.eventually.intro}]
By Lemma~\ref{l.unaryoutput} we may assume that the
processor $\Proc$ to be emulated has unary output.  Suppose
that $\Proc$ computes $F: \N^k \to \N$, and recall from
\cref{t.eventually} that $F$ is ZILEP.  We will use
induction on $k$, with Lemma~\ref{unary-trans} providing
the base case.  We therefore suppose that $k\geq 2$ and
focus on the $k$th coordinate. Suppose that
\begin{equation}\label{def-slr}
F(\yy,z+L)=F(\yy,z)+SL\quad\text{for all }\yy\in\N^{k-1}\text{ and }z\geq
R.
\end{equation}
We call $L$ the \df{period}, $S$ the \df{slope}, and $R$ the \df{margin} (the
width of the non-periodic part) of $F$ with respect to the $k$th coordinate.
(In the notation of \cref{s.recurrent}, one choice is to take $L=\lambda_k$, $S=b_k$ and
$R=r_k$. Note that $L$ is defined only up to positive integer multiples, and $R$ can always be increased. On the other hand, $S$ is uniquely defined.)
 Motivated by the proof of \cref{baby-case}, we also consider the parameter
\begin{equation}\label{def-w}
W:=\sup_{(\yy,z)\in \N^{k}} F(\yy,z+1)-F(\yy,z),
\end{equation}
which we call the \df{roughness} of $F$.  Since the difference inside the
supremum is an eventually periodic function of $(\yy,z)$, we have $W<\infty$.

If $W=0$ then $F$ does not depend on the $k$th coordinate, so we are done by
induction.  Assuming now that $W>0$, we will first reduce to functions satisfying
\eqref{def-slr} and \eqref{def-w} with parameters satisfying one of the
following:

\begin{center}
\begin{tabular}{llll}
	\textbf{Case 0:} & $W=1$,& $L=R=n$, & $S=0$; \\
	\textbf{Case 1:} & $W=1$,& $L=R=n$, & $S=1/n$,
\end{tabular}
\end{center}
where in both cases, $n$ is a positive integer.
\medskip

\paragraph{\bf Reduction to Case 0.}  Suppose that the original function $F$ has slope
$S=0$.  We use the meagerization identity \eqref{e.meagerization} to express
$F$ as the sum $\sum_{j=0}^{W-1} F_j$ where $F_j= \floor{(F+j)/W}$. Each
$F_j$ is ZILEP by \cref{backwards} or \cref{t.eventually}.  We will emulate
each $F_j$ separately and use an adder.  Note that $F_j$ has roughness $1$.
Moreover, since $S=0$ we have
	\begin{equation} \label{e.poopsout} F_j(\yy,z)= F_j(\yy,R) \quad\text{for
all }z\geq R.\end{equation}
Therefore we can set $n=R$,
and $F_j$ satisfies \eqref{def-slr} and \eqref{def-w} with the claimed
parameters for case 0.

\medskip
\paragraph{\bf Reduction to Case 1.} Suppose on the other hand that $F$ has slope $S>0$.
Take $R$ larger if necessary so that $R\geq WL$. We
use the meagerization identity \eqref{e.meagerization} to express $F$ as
$\sum_{j=0}^{Sn-1} F_j$ where
$$n := \Bigl\lfloor\frac{R}{WL}\Bigr\rfloor WL; \qquad
F_j:=\Bigl\lfloor
\frac{F+j}{Sn}\Bigr\rfloor.$$
Note that $n \geq 1$.  We will again emulate each $F_j$ separately and add
them. Note that since $Sn \geq \max(R,SL)$, each $F_j$ has roughness $1$.
Next we check that $n$ can be taken as both the period and the margin for
each $F_j$. Since $L$ divides $n$ and $n \geq R$, we have
	\[ F_j (\yy,z{+}n) = \floor{\frac{F(\yy,z) + Sn + j}{Sn}} = F_j(\yy,z) +1
\] for all $\yy \in \N^{k-1}$ and all $z \geq n$.  Finally, $F_j$ has slope
$1/n$ as desired.

\medskip
\paragraph{\bf Inductive step.}
Assume now that $F$ satisfies \eqref{def-slr} and \eqref{def-w} with
parameters as given in case 0 or case 1 in the above table.  Also assume
that the result of the theorem holds for all processors with $k-1$ inputs.
We will apply the interleaving result, \cref{l.stronginterleaving}, rewritten in terms of functions
that map $\zero$ to $0$.  To this end, define for $i=0,\ldots, n-1$,
\begin{align*}
  \zeta_i(z)&:=\Bigl\lfloor\frac{z+n-i-1}{n}\Bigr\rfloor, &z\in\N;\\
  u_i&:=F(\zero,i);& \\
  G_i(\yy,\zeta)&:=F(\yy,n\zeta+i)-u_i, &\yy\in \N^{k-1},\;\zeta\in\N;
\intertext{and}
  M_\uu(\vv)&:=M(\vv+\uu), & \vv\in\N^n,
\end{align*}
where $\uu:=(u_0,\ldots,u_{n-1})\in \Z^n$ and $M$ is the function from
\cref{p.blackbox}. Then we have
\begin{equation}\label{central-equation}
 F(\yy,z)=M_\uu\Bigl(
G_0\bigl(\yy,\zeta_0(z)\bigr),\ldots,G_{n-1}\bigl(\yy,\zeta_{n-1}(z)\bigr)\Bigr).
\end{equation}
This follows
 from \cref{l.stronginterleaving}: the condition \eqref{e.benign} is satisfied because
 $W=1$.

Note that $u_0=0$ and $u_{i+1}-u_{i}\leq 1$ (also because
$W=1$), so by \cref{p.blackbox} we have
$M_\uu(\zero)=M(\uu)=0$.  Moreover, the function $M_\uu$ is
increasing and periodic, so by \cref{t.recurrent} there is
an recurrent abelian processor $\ProcM_\uu$ that computes
it, and by \cref{t.recurrent.intro} there exists a network
(of topplers, adders and splitters) that emulates
$\ProcM_\uu$.  Also note that the function $z\mapsto
\zeta_i(z)$ is computed by a primed toppler.

For each $i$, the function $G_i$ is increasing, and can be expressed as a
linear plus an eventually periodic function (by \cref{t.eventually}). And we have
$$G_i(\zero,0)=F(\zero,i)-u_i=0.$$
So our task is reduced to finding a network to compute $G_i$.  For all $\yy$
and $\zeta\geq 1$ we have
$$G_i(\yy,\zeta+1)-G_i(\yy,\zeta)=F(\yy,n\zeta+i+n)-F(\yy,n\zeta+i)
=\begin{cases}
  0&\text{in case 0}\\
  1&\text{in case 1}.
\end{cases}
$$

Thus, in case 0, $G_i$ is ZILEP and satisfies the condition
\eqref{constant-after-1}, so by \cref{baby-case} it can be computed by a
network of gates and $(k-1)$-ary processors.  By the inductive hypothesis,
each $(k-1)$-ary processor can be replaced with a network of gates that
emulates it.

On the other hand, in case 1 we can write
$$G_i(\yy,\zeta)=H_i(\yy,\zeta)+(\zeta-1)^+,$$
for a function $H_i$ (which can be written
$H_i(\yy,\zeta):=G_i(\yy,\ind[\zeta>0])$), that is ZILEP
and satisfies \eqref{constant-after-1}.  By
\cref{baby-case} and the inductive hypothesis, $H_i$ can be
computed by a network of gates.  Thus, we can compute $G_i$
by feeding $\zeta$ into a splitter, sending one output to a
delayer and the other to a processor $\ProcH_i$ that
computes $H_i$, and adding the results.  See \cref{case1}.
\begin{figure}
\centering
\setlength{\tabcolsep}{8pt}
\begin{tabular}{ccc}
\begin{tikzpicture}[abelian,baseline=-1mm]
  \node[block,minimum height=1cm,minimum width=1cm] (h) {$\ProcG_i$};
\end{tikzpicture}
&$\Longleftrightarrow$&
\begin{tikzpicture}[abelian,baseline=-1cm]
  \node[block,minimum height=1cm,minimum width=1cm] (h) {$\ProcH_i$};
  \node[delayer,below of=h](d){};
  \path (h) -- (d) node[pos=.5,right=1.5cm,adder] (a){};
  \node[splitter,left of=d](s){};
  \draw[<-] (s) -- ++(-1,0) node[left](z){$\zeta$};
  \node[above of=z] (y){$\yy$};
  \draw[->] (s) to (h);
  \draw[->] (y) to (h);
  \draw[->] (s) to (d);
  \draw[->] (h) to[bend left=20] (a);
  \draw[->] (d) to[bend right=20] (a);
  \draw[->] (a)--++(1,0);
\end{tikzpicture}
\end{tabular}
\caption{The additional reduction in case 1.}\label{case1}
\end{figure}
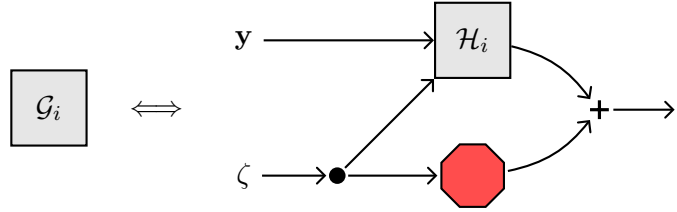

Finally, \eqref{central-equation} and \cref{fig:induct}
show how to complete the emulation of $F$ in either case:
$z$ is split and fed into various primed topplers that
compute the functions $\zeta_i(z)$, which are combined with
$\yy$ and fed into networks emulating the $\ProcG_i$. The
results are combined using $\ProcM_\uu$.
\begin{figure}
\centering
\begin{tikzpicture}[abelian]
  \node[block,minimum height=1cm,minimum width=1cm] (p0) {$\ProcG_0$};
  \node[toppler,below left of=p0,xshift=-3mm,inner sep=0pt,prime=$n-1$] (t0){$n$};
  \node[block,below of=p0,minimum height=1cm,minimum width=1cm] (p1) {$\ProcG_1$};
  \node[toppler,below left of=p1,xshift=-3mm,inner sep=0pt,prime=$n-2$] (t1){$n$};
  \node[below=2mm of p1] (ddd){$\vdots$};
  \node[block,below= 5mm of ddd,minimum height=1cm,minimum width=1cm] (pn) {$\ProcG_n$};
  \node[toppler,below left of=pn,xshift=-3mm,prime=$0$] (tn){$n$};
  \node[block,right=1.3cm of p1,minimum height=2cm,minimum width=1cm](m){$\ProcM_\uu$};
  \node[splitter,left=4cm of p1] (sy){};
  \node[splitter,left=4cm of pn] (sz){};
  \draw[<-] (sy)--++(-1,0)node[left] (y){$\yy$};
  \draw[<-] (sz)--++(-1,0)node[left] (z){$z$};
  \draw[->] (m)--++(1,0);
  \draw[->] (sy)to[bend left=30](p0);
  \draw[->] (sy)to(p1);
  \draw[->] (sy)to[bend right=30](pn);
  \draw[->] (sz)to[bend left=30](t0.west);
  \draw[->] (sz)to(t1.west);
  \draw[->] (sz)to[bend right=30](tn.west);
  \draw[->] (p0)to(m);
  \draw[->] (p1)--(m);
  \draw[->] (pn)to(m);
  \draw[->] (t0)--(p0);
  \draw[->] (t1)--(p1);
  \draw[->] (tn)--(pn);
\end{tikzpicture}
\caption{The main inductive step.}\label{fig:induct}
\end{figure}
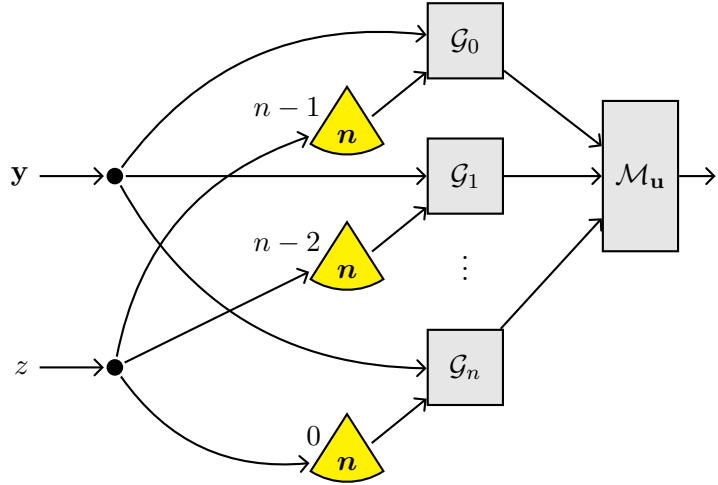
\end{proof}

\section{Necessity of all gates}
\label{s.lowerbounds}

In this section we study the classes of functions computable by various
subsets of the abelian logic gates in Table~\ref{fig:gates}. The following
observation will be useful: A function on $\N^k$ can be decomposed in at most
one way as the sum of a linear and an eventually periodic function. Indeed,
the difference of two linear functions is either zero or unbounded on $\N^k$,
so if
	\[ L_1 + P_1 = L_2 + P_2 \] for some linear functions $L_1, L_2$ and some
eventually periodic functions $P_1, P_2$, then $L_1 - L_2$ is bounded on
$\N^k$ and hence $L_1 = L_2$, which in turn implies $P_1 = P_2$.

\subsection{Necessity of infinitely many component types}

We have seen that $2$-topplers, splitters and adders suffice to emulate any
finite recurrent abelian processor if feedback is permitted. The goal of this
section is to show that no finite list of components will suffice to emulate
all finite recurrent processors by a \emph{directed acyclic} network.

The \textbf{exponent} of a recurrent abelian processor is the smallest
positive integer $m$ such that inputting $m$ copies of any letter acts as the
identity: $t_i^m(q)=q$ for all recurrent states $q$ and all input letters
$i$.

\begin{lemma}
\label{l.exponent1} Let $\Net$ be a finite, directed acyclic network of
recurrent abelian components. The exponent of $\Net$ divides the product of
the exponents of its components.
\end{lemma}

\begin{proof}
Induct on the number of components. Since $\Net$ is directed acyclic, it has
at least one component $\Proc$ such that no other component feeds into
$\Proc$. Let $m$ be the exponent of $\Proc$, and let $M$ be the product of
the exponents of all components of $\Net$. For any letter $i$ in the input
alphabet of $\Proc$, if we input $M$ letters $i$ to $\Proc$, then $\Proc$
returns to its initial recurrent state and outputs a nonnegative integer
multiple of $M/m$ letters of each type. By induction, the exponent of the
remaining network $\Net - \Proc$ is divisible by $M/m$, so all other
processors also return to their initial recurrent states.
\end{proof}

\old{ 
\begin{remark}
The product of exponents in the above lemma can be improved slightly to the quantity
	\[ M = \mathrm{lcm} \{ m_\gamma | \gamma : I \to O \}, \] the least
common multiple over all paths $\gamma$ from an input node to an output node,
of the product $m_\gamma$ of the exponents of the components along $\gamma$.
\end{remark}

\begin{remark}
Lemma~\ref{l.exponent1}
remains true even if we do not require the components to be recurrent.
The proof in this case uses that any recurrent state of $\Net$ is also
locally recurrent \cite[Lemma 3.9]{BL15c}. 
\end{remark}
}

\begin{lemma}
\label{l.exponent2} Let $\Net$ be a finite, directed acyclic network of
recurrent abelian components that emulates a $\lambda$-toppler. Then
$\lambda$ divides the exponent of $\Net$.
\end{lemma}

\begin{proof}
If $m$ is the exponent of $\Net$, then $x \mapsto F_{\Net}(mx)$ is a linear
function. Equating the linear parts of the $L+P$ decomposition of $\Net$ and
the $\lambda$-toppler, we obtain
	\[ \frac{F_{\Net}(mx)}{m} = \frac{x}{\lambda} \]
for all $x \in \N$. Setting $x=1$ gives $\lambda$ divides $m$.
\end{proof}

Lemmas~\ref{l.exponent1} and~\ref{l.exponent2} immediately imply the following.

\begin{corollary}
Let $\mathcal{L}$ be any finite list of finite recurrent abelian processors.
There exists $p \in \N$ such that a finite, directed acyclic network of
components from $\mathcal{L}$ cannot emulate a $p$-toppler.
\end{corollary}

\begin{proof}
Let $p$ be a prime that does not divide the exponent of any member of $\mathcal{L}$.
\end{proof}

\subsection{Necessity of primed topplers in the recurrent case.}

A directed acyclic network of adders, splitters and unprimed topplers
computes a function $L+P$ with $L$ linear and $P$ periodic with $P \leq 0$.
The inequality follows from converting each toppler $\floor{x/\lambda}$ into
its linear part $x/\lambda$.  Recall however that we can do away with primed
topplers if we allow presinks (\cref{l.unpriming}).

\subsection{Necessity of delayers and presinks.}

Lemma~\ref{l.acyclic} implies that a directed acyclic network of recurrent
components is itself recurrent, so at least one transient gate is needed in
order to emulate an arbitrary finite abelian processor. But why do we have
\emph{two} transient gates, the delayer and the presink?  In this section we
will show that neither can be used along with recurrent components to emulate
the other.

\begin{lemma}
\label{l.zilpandbounded}
If $G : \N \to \N$ is both ZILP and bounded, then $G \equiv 0$.
\end{lemma}

\begin{proof}
Write $G = L+P$ for $L$ linear and $P$ periodic. In particular, $P$ is
bounded, so if $G$ is bounded then $L$ is both bounded and linear, hence
zero. But then $G=P$, and the only increasing periodic function is the zero
function.
\end{proof}

\begin{prop}
Let $\Net$ be a finite directed acyclic network of recurrent components and
delayers. Then $\Net$ cannot emulate a presink.
\end{prop}

\begin{proof}
Let $A$ be the total alphabet of $\Net$, and let $F = F_\Net : \N^A \to \N$.
Let $D \subset A$ the set of incoming edges to the delayers. Note that
inputting $\11_D$ converts all delayers to wires, and has no other effect (in
particular, no output is produced: $F(\11_D)=0$). The resulting network with
delayers converted to wires is recurrent by Lemma~\ref{l.acyclic}, so the
function
	\[  \til F (\xx) := F(\xx + \11_D) \]
is ZILP by Theorem~\ref{t.recurrent}.

Now suppose for a contradiction that $\Net$ emulates a presink; that is, for
some letter $a \in A$ we have $F(n \11_a) = \one \{n > 0\}$. Then the
function
	\[ G(n) := F(n \11_a + \11_D) \] is bounded (by $1+ \max_q
F_{\Net,q}(\11_D)$, where the maximum is over the finitely many states $q$ of
$\Net$). Since $G$ is the restriction of the ZILP function $\til F$ to a
coordinate ray, $G$ is ZILP, which implies $G \equiv 0$ by
Lemma~\ref{l.zilpandbounded}.  But $G(1) \geq F(\11_a) = 1$, which gives the
required contradiction.
\end{proof}

The proof shows a bit more: If $\Net$ is a directed acyclic network of
recurrent components and delayers, then $F_\Net$ is either zero or unbounded
along any coordinate ray.

\begin{lemma}
\label{l.zilpequalslinear} If $G : \N \to \N$ is ZILP, say $G=L+P$ with $L$
linear and $P$ periodic, then $G(x) = L(x)$ for infinitely many $x$.
\end{lemma}

\begin{proof}
Since $G(0)=L(0)=0$ we have $P(0)=0$. Since $P$ is periodic, $P(x)=0$ for infinitely many $x$.
\end{proof}

\begin{prop}
Let $\Net$ be a finite, directed acyclic network of
recurrent components and presinks. Then $\Net$ cannot
emulate a delayer.
\end{prop}

\begin{proof}
Let $A$ be the total alphabet of $\Net$, and let $F =
F_\Net : \N^A \to \N$.  Let $S \subset A$ the set of
incoming edges to the presinks. Note that inputting $\11_S$
converts all presinks to sinks. However, unlike the input
$\11_D$ of the previous proposition, the input $\11_S$ may
have other effects: It may change the states of other
components, and may produce a nonzero output $F(\11_S)$.

Denote by $\qq^0$ the initial state of $\Net$ and by
$\qq^1$ the state resulting from input $\11_S$. The
resulting network $\ProcR$ with presinks converted to sinks
is recurrent by Lemma~\ref{l.acyclic}, so the function
	\[  F_{\ProcR, \qq^1} (\xx) = F(\xx + \11_S) - F(\11_S) \]
is ZILP by Theorem~\ref{t.recurrent}.

Now we relate $F_{\ProcR, \qq^1}$ to $F_{\ProcR, \qq^0}$.
Since $\ProcR$ is recurrent, there is an input $\uu \in
\N^A$ such that inputting $\uu$ to $\ProcR$ in state
$\qq^1$ results in state $\qq^0$. Since converting presinks
to sinks \emph{without} changing the states of any other
components cannot increase the output, we have
	\begin{align*} F(\xx) = F_{\Net,\qq^0}(\xx) &\geq F_{\ProcR, \qq^0}(\xx) \\
		&= F_{\ProcR, \qq^1}(\xx+\uu) - F_{\ProcR, \qq^1}(\uu) \\
		&= F(\xx+\uu+\11_S) - F(\uu+\11_S). \end{align*}

Finally, suppose for a contradiction that $\Net$ emulates a
delayer; that is, for some letter $a \in A$ we have $F(n
\11_a) = (n-1)^+$. Then the function
	\[ G(n) := F(n \11_a + \uu+\11_S) - F(\uu+\11_S), \] is
ZILP with linear part $L(n)=n$.  By
Lemma~\ref{l.zilpequalslinear}, $G(n)=n$ for infinitely
many $n$.  This yields the required contradiction, since $n
> F(n \11_a) \geq G(n)$ for all $n \geq 1$.
\end{proof}

\old{ A directed acyclic network of recurrent components
and delayers, if it sends any output at all, must be
nondegenerate: if it computes $f(x) = Lx + P(x)$ then
either $f \equiv 0$ or $L>0$. The proof is by induction on
the number of components: if some component sends output
then it must receive input, so by the inductive hypothesis
it can be made to receive an arbitrarily large input and
hence send an arbitrarily large output.  So a presink
cannot be emulated by such a network. }

\section{Open problems}
\label{s.open}

\subsection{Floor depth}

Let us define the \emph{floor depth} of a ZILP function
as the minimum number of nested floor functions in a formula for it. More
precisely, let $\mathcal{R}_0$ be the set of $\N$-affine functions $\N^k \to
\N$, and for $n \geq 1$ let $\mathcal{R}_n$ be the smallest set of functions
closed under addition and containing all functions of the form
$\floor{f/\lambda}$ for $f \in \mathcal{R}_{n-1}$ and positive integer
$\lambda$.  The floor depth of $f$ is defined as the smallest $n$ such that
$f \in \mathcal{R}_n$.

If $f$ is computed by a directed acyclic network of splitters, adders and
topplers, then the proof of Corollary~\ref{t.functions} in
Section~\ref{s.splitteradder} shows that the floor depth of $f$ is at most
the maximum number of topplers on a directed path in the network. Hence, by
the construction of the emulating network in Section~\ref{s.recurrent}, every
ZILP function $\N^k \to \N$ has floor depth at most $k$.
Is this sharp?

\subsection{Unprimed topplers}

What class of functions $\N^k \to \N$ can be computed by a
directed acyclic network of splitters, adders and unprimed
topplers?


\subsection{Conservative gates}

Call a finite abelian processor \emph{conservative} if, in the matrix of the
linear part of the function it computes, each column sums to $1$.  We can
think of the input and output letters of such a processor as
indistinguishable physical objects (\emph{balls}) that are conserved,
as in a sandpile or rotor-router model with no sources or sinks.
An internal state represents a configuration of (a bounded number of) balls
stored inside the processor. Splitters and topplers are not conservative:
splitters create balls while topplers consume them.
But a sandpile node -- which distributes $k$ particles each time it receives $k$ -- is conservative,
even though we would ordinarily emulate one using a toppler and splitters.
A finite network of
conservative abelian processors with no trash edges emulates a single
conservative processor (provided the network halts). Find a minimal set of
conservative gates that allow any finite conservative abelian processor to be
emulated.

\subsection{Gates with infinite state space}

Each of the following functions $\N^2 \to \N$
	\begin{align*}
 	(x,y) &\mapsto \min(x,y) \\
	 (x,y) &\mapsto \max(x,y) \\
	 (x,y) &\mapsto xy \end{align*} can be computed by an abelian processor
with an infinite state space. In the case of $\min$ and $\max$ the state
space $\N$ suffices, with transition function $t_{(x,y)}(q) = q+x-y$. The
product $(x,y) \mapsto xy$ requires state space $\N^2$,
as well as unbounded output: for example, when it receives input $\basis_1$
in state $(x,y)$ it transitions to state $(x+1,y)$ and outputs $y$ letters.
What class of functions can be computed by an abelian network (with or
without feedback) whose components are finite abelian processors and a
designated subset of the above three? Such functions have an $L+P$
decomposition where the $L$ part is piecewise linear, polynomial or piecewise
polynomial (Table~\ref{table.L+P}). \old{ Integrality becomes an interesting
constraint in this setting. For example, the function $x \mapsto (x^2+x)/2$
is $\N$-valued for $x \in \N$. It can obviously be computed using a
$2$-toppler, adder and multiplier, but can it be computed by a network of
immutable components and multipliers? }

	\begin{table}
	\begin{center}
	\begin{tabular}{|c|c|}
    \hline
	L & P \\
	\hline
	linear & zero \\
	piecewise linear & eventually constant \\
	polynomial & periodic \\
	piecewise polynomial & eventually periodic \\
\hline
	\end{tabular}
	\end{center}
	\medskip
	\caption{Sixteen ($4\times 4$) types of $L+P$ decomposition
for an increasing function $\N^k \to \N^\ell$.}
	\label{table.L+P}
	\end{table}

\section*{Acknowledgments}

We thank Ben Bond, Sergey Fomin and Jeffrey Lagarias for
inspiring conversations, and Swee Hong Chan for carefully reading an early draft.


\begin{thebibliography}{9}
\bibitem[Cai15]{Cai15} Hannah Cairns, Some halting problems for abelian sandpiles are undecidable in dimension three. \arxiv{1508.00161}
\end{thebibliography}
\end{document}